\pdfoutput=1 %
\documentclass[11pt,letter]{article}
\usepackage{microtype}%
\usepackage[in]{fullpage}

\usepackage[utf8]{inputenc}
\usepackage[pdftex]{graphicx}
\usepackage[usenames,dvipsnames,table]{xcolor}

\usepackage{enumitem}

\usepackage{amssymb}
\usepackage{amsmath}
\usepackage{amsthm}%
\usepackage{mathtools}%
\usepackage{thmtools}
\usepackage{thm-restate}

\usepackage{float}
\usepackage{morefloats}
\usepackage{cite} %

\usepackage{multirow}

\usepackage{adjustbox}
\usepackage{tikz}
\usetikzlibrary{positioning}

\usepackage[noend]{algpseudocode} %

\usepackage{hyperref}

\usepackage{url}
\urldef{\maildiku}\path|roden@diku.dk|
\urldef{\mailjacob}\path|jaho@di.ku.dk|

\usepackage[textwidth=0.8in,textsize=scriptsize]{todonotes}

\DeclareMathOperator*{\argmin}{arg\,min}
\DeclareMathOperator*{\argmax}{arg\,max}

\newcommand{\Oo}{\mathcal O} %
\newcommand{\oo}{o} %

\newcommand{\abs}[1]{\left\lvert#1\right\rvert}
\newcommand{\ceil}[1]{\left\lceil#1\right\rceil}
\newcommand{\floor}[1]{\left\lfloor#1\right\rfloor}
\newcommand{\set}[1]{\left\{#1\right\}}
\newcommand{\cond}{\mathrel{}\middle\vert\mathrel{}}
\newcommand{\sizeof}[1]{\left\lvert#1\right\rvert}

\newcommand{\meet}{\operatorname{meet}}
\newcommand{\diag}{\operatorname{diag}}
\newcommand{\dist}{\operatorname{dist}}

\renewcommand{\root}{\operatorname{root}} %
\newcommand{\vertex}{\operatorname{vertex}}

\newcommand{\nil}{\mathbf{nil}}
\newcommand{\true}{\mathbf{true}}
\newcommand{\false}{\mathbf{false}}

\newcommand{\boundary}{\partial}

\newcommand{\Cover}{\operatorname{Cover}}
\newcommand{\Uncover}{\operatorname{Uncover}}
\newcommand{\Expose}{\operatorname{Expose}}
\newcommand{\Merge}{\operatorname{Merge}}
\newcommand{\Split}{\operatorname{Split}}

\newcommand{\CleanToBuffer}{\operatorname{CleanToBuffer}}
\newcommand{\CleanToChild}{\operatorname{CleanToChild}}
\newcommand{\ComputeFromBuffer}{\operatorname{ComputeFromBuffer}}
\newcommand{\ComputeFromChild}{\operatorname{ComputeFromChild}}

\newcommand{\cover}{\operatorname{cover}}
\newcommand{\globalcover}{\operatorname{globalcover}}

\newcommand{\packedcover}{\operatorname{packedcover}}
\newcommand{\packedglobalcover}{\operatorname{packedglobalcover}}

\newcommand{\minpathedge}{\operatorname{minpathedge}}
\newcommand{\minglobaledge}{\operatorname{minglobaledge}}

\newcommand{\pointset}{\operatorname{pointset}}
\newcommand{\pointsize}{\operatorname{pointsize}}
\newcommand{\pointincident}{\operatorname{pointincident}}

\newcommand{\size}{\operatorname{size}}
\newcommand{\partpath}{\operatorname{partpath}}
\newcommand{\partsize}{\operatorname{partsize}}
\newcommand{\diagsize}{\operatorname{diagsize}}
\newcommand{\partsizesum}{\operatorname{partsizesum}}
\newcommand{\diagsizesum}{\operatorname{diagsizesum}}
\newcommand{\incident}{\operatorname{incident}}
\newcommand{\partincident}{\operatorname{partincident}}
\newcommand{\diagincident}{\operatorname{diagincident}}

\newcommand{\tree}{\operatorname{tree}}
\newcommand{\undo}{\operatorname{undo}}

\newcommand{\lmax}{\ell_{\max}}
\newcommand{\CoverLevel}{\operatorname{CoverLevel}}
\newcommand{\MinCoveredEdge}{\operatorname{MinCoveredEdge}}
\newcommand{\FindSize}{\operatorname{FindSize}}

\newcommand{\AddLabel}{\operatorname{AddLabel}}
\newcommand{\RemoveLabel}{\operatorname{RemoveLabel}}
\newcommand{\FindFirstLabel}{\operatorname{FindFirstLabel}}

\newcommand{\Patrascu}{P\v{a}tra\c{s}cu}

\declaretheorem[style=plain,name={Theorem}]{theorem}
\declaretheorem[style=plain,name={Lemma},sibling=theorem]{lemma}
\declaretheorem[style=plain,name={Corollary},sibling=theorem]{corollary}

\usepackage{authblk}

\author{Jacob Holm\thanks{This research is supported by Mikkel Thorup's Advanced Grant DFF-0602-02499B from the Danish Council for Independent Research under the Sapere Aude research career programme.}}
\author{Eva Rotenberg}
\author{Mikkel Thorup$^\ast$}
\affil{University of Copenhagen (DIKU),\\
    jaho@di.ku.dk, eva@rotenberg.dk, mthorup@di.ku.dk} 

\title{%
    Dynamic Bridge-Finding in $\widetilde{O}(\log ^2 n)$ Amortized Time}

\begin{document}

\thispagestyle{empty}
\maketitle
\begin{abstract}
  We present a
  deterministic
  fully-dynamic data structure
  for maintaining information about the bridges in a graph.
  We support updates in $\widetilde{O}((\log n)^2)$ amortized time, and can find a bridge in the component of any given vertex, or a bridge separating any two given vertices, in $\Oo(\log n / \log \log n)$ worst case time.
  Our bounds match the current best for bounds for deterministic 
  fully-dynamic connectivity up to $\log\log n$ factors.

  The previous best dynamic bridge finding was an 
  $\widetilde{O}((\log n)^3)$
  amortized time algorithm by Thorup [STOC2000], which was a bittrick-based improvement on the $\Oo((\log n)^4)$ amortized time algorithm by Holm et al.[STOC98, JACM2001].

  Our approach is based on a different and purely combinatorial
  improvement of the algorithm of Holm et al., 
 which by itself gives a new combinatorial $\widetilde{O}((\log n)^3)$
  amortized time algorithm. Combining it with Thorup's bittrick, 
  we get down to the claimed $\widetilde{O}((\log n)^2)$
  amortized time. 

  Essentially the same new trick can be applied to the biconnectivity data structure from [STOC98, JACM2001], improving the amortized update time to $\widetilde{O}((\log n)^3)$.

  We also offer improvements in space. We describe a general
  trick which applies to both of our new algorithms, and to the old
  ones, to get down to linear space, where the previous best use
  $\Oo(m + n\log n\log\log n)$.

  Our result yields an improved running time for deciding whether a unique perfect matching exists in a static graph.%
\end{abstract}

 \setcounter{page}{0}
\thispagestyle{empty}
\newpage
\section{Introduction}
In graphs and networks, connectivity between vertices is a fundamental property. In real life, we often encounter networks that change over time, subject to insertion and deletion of edges. We call such a graph \emph{fully dynamic}. Dynamic graphs call for dynamic data structures that maintain just enough information about the graph in its current state to be able to promptly answer queries. 

Vertices of a graph are said to be \emph{connected} if there exists a path between them, and \emph{$k$-edge connected} if no sequence of $k-1$ edge deletions can disconnect them. A \emph{bridge} is an edge whose deletion would disconnect the graph.
In other words, a pair of connected vertices are $2$-edge connected if they are not separated by a bridge.  By Menger's Theorem~\cite{menger1927allgemeinen}, this is equivalent to saying that a pair of connected vertices are $2$-edge connected if there exist two edge-disjoint paths between them. By edge-disjoint it is meant that no edge appears in both paths.

For dynamic graphs, the first and most fundamental property to be studied was that of dynamic connectivity. In general, we can assume the graph has a fixed set of $n$ vertices, and we let $m$ denote the current number of edges in the graph. The first data structure with sublinear $\Oo(\sqrt n)$ update time is due to Frederickson~\cite{DBLP:journals/siamcomp/Frederickson85} and Eppstein et al.~\cite{Eppstein93improvedsparsification}%
. Later, Frederickson~\cite{DBLP:journals/siamcomp/Frederickson97} and Eppstein et al.~\cite{Eppstein93improvedsparsification} gave a data structure with $\Oo(\sqrt n)$ update time for $2$-edge connectivity%
.
Henzinger and King achieved poly-logarithmic expected amortized time~\cite{Henzinger1997}, that is, an expected amortized update time of $\Oo((\log n)^3)$, and $\Oo(\log n /\log \log n)$ query time for connectivity.  And in~\cite{Henzinger97fullydynamic}, $\Oo((\log n)^5)$ expected amortized update time and $\Oo(\log n)$ worst case query time for $2$-edge connectivity.
The first polylogarithmic deterministic result was by Holm et al. announced in~\cite{Holm:1998}, see \cite{Holm:2001} for a journal version; an amortized deterministic update time of $\Oo((\log n)^2)$ for connectivity, and $\Oo((\log n)^4)$ for $2$-edge connectivity. The update time for deterministic dynamic connectivity has later been improved to $\Oo((\log n)^2 /\log \log n)$ by Wulff-Nilsen~\cite{Wulff-Nilsen16}.
Sacrificing determinism, an $\Oo(\log n(\log\log n)^3)$ structure for connectivity was presented by Thorup~\cite{Thorup:2000}, and later improved to  $\Oo(\log n (\log\log n)^2)$ by Huang et al.~\cite{Huang:2017}.
In the same paper, Thorup obtains an update time of $\Oo((\log n)^3 \log \log n)$ for deterministic $2$-edge connectivity. 
Interestingly, Kapron et al.~\cite{Kapron:2013} gave a Monte Carlo-style randomized data structure with polylogarithmic worst case update time for dynamic connectivity, namely, $\Oo((\log n)^4)$ per edge insertion, $\Oo((\log n)^5)$ per edge deletion, and $\Oo(\log n/\log\log n)$ per query. This was later improved by Gibbs et al.~\cite{DBLP:journals/corr/GibbKKT15} to $\Oo((\log n)^4)$ worst case update time and sublinear $O(n\log^2 n)$ space.  We know of no similar worst-case result for bridge finding.
The same paper~\cite{DBLP:journals/corr/GibbKKT15} also gives the first sublinear-space $O(n\log^2 n)$ space data structure for (amortized)
$2$-edge connectivity, by using the sublinear-space connectivity
data structure to maintain a sparse subgraph preserving $2$-edge
connectivity and then using the existing $2$-edge connectivity
data structure from Holm et al.~\cite{Holm:2001} as a black box on that
subgraph.

The best lower bound known is by \Patrascu{} et al.~\cite{patrascu2006logarithmic},
which shows a trade-off between update time $t_u$ and query time $t_q$ of $t_q \lg\frac{t_u}{t_q} = \Omega (\lg n)$ and $t_u \lg\frac{t_q}{t_u} = \Omega (\lg n)$.

\subsection{Our results}
We obtain an update time of $\Oo((\log n)^2(\log\log n)^2)$ and a query time of $\Oo(\log n /\log \log n)$ for the bridge finding problem:
\begin{theorem} \label{thm:MainThm}
    There exists a deterministic data structure for dynamic multigraphs in the word RAM model with $\Omega(\log n)$ word size, that uses $\Oo(m+n)$ space, and can handle the following updates, and queries for arbitrary vertices $v$ or arbitrary connected vertices $v,u$: 
    \begin{itemize}
        \item insert and delete edges in $\Oo((\log n)^2(\log\log n)^2)$ amortized time,
        \item %
        find a bridge in $v$'s connected component or determine that none exists, %
        or find a bridge that separates $u$ from $v$ or determine that none exists. Both in $\Oo(\log n /\log \log n)$ worst-case time. %
        \item %
        find the size of $v$'s connected component in $\Oo(\log n/\log\log n)$ worst-case time, or the size of its $2$-edge connected component in $\Oo(\log n(\log\log n)^2)$ worst-case time.

    \end{itemize}
\end{theorem}
Since a pair of connected vertices are $2$-edge connected exactly when there is no bridge separating them, we have the following corollary:
\begin{corollary}
    There exists a data structure for dynamic multigraphs in the word RAM model with $\Omega(\log n)$ word size, that can answer $2$-edge connectivity queries in $\Oo(\log n/\log\log n)$ worst case time and handle insertion and deletion of edges in $\Oo((\log n)^2(\log\log n)^2)$ amortized time, with space consumption $\Oo(m+n)$.
\end{corollary}
Note that the query time is optimal with respect to the trade-off by \Patrascu{} et al.~\cite{patrascu2006logarithmic}

As a stepping stone on the way to our main theorem, we show the following:
\begin{theorem}\label{thm:Combinatorial}
   There exists a combinatorial deterministic data structure for dynamic multigraphs on the pointer-machine without the use of bit-tricks, that uses $\Oo(m+n)$ space, and can handle insertions and deletions of edges in $\Oo((\log n)^3\log\log n)$ amortized time, find bridges and determine connected component sizes in $\Oo(\log n )$ worst-case time, and find $2$-edge connected component sizes in $\Oo((\log n)^2\log\log n)$ worst-case time.
\end{theorem}

Our results are based on modifications to the $2$-edge connectivity data structure from~\cite{Holm:2001}.  Applying the analoguous modification to the biconnectivity data structure from the same paper yields a structure with $\Oo((\log n)^3(\log\log n)^2)$ amortized update time and $\Oo((\log n)^2(\log\log n)^2)$ worst case query time.  The details of this modification are beyond the scope of this paper.  %

\subsection{Applications}\label{sec:applications}
Although our data structure is deterministic and uses linear space, it entails an improvement of the current best sublinear-space data structure. 
Namely, the Monte-Carlo randomized sublinear-space $2$-edge connectivity data structure 
by Gibbs et al.~\cite{DBLP:journals/corr/GibbKKT15} uses the data structure 
from~\cite{Holm:2001} as a black box:  For each update the data structure uses worst case $O(\log^5 n)$ time by itself, and makes $O(\log^2 n)$ updates in the sparse graph seen by the black box. Thus, with the $O(\log^4 n)$ amortized update time from~\cite{Holm:2001}, this gives a sublinear-space data structure with amortized $O(\log^6 n)$ update time. Using the data structure from~\cite{Thorup:2000} or our new purely combinatorial data structure, this drops to $\Oo((\log n)^5\log\log n)$ amortized time. With our new $\tilde{O}(\log ^2 n)$ update time data structure, this improves to $\Oo(\log^5 n)$ amortized time (and the bottleneck is now in the reduction).

While dynamic graphs are interesting in their own right, many algorithms and theorem proofs for static graphs rely on decremental or incremental graphs. Take for example the problem of whether or not a graph has a unique perfect matching. The following theorem by Kotzig immediately yields a near-linear time algorithm if implemented together with a decremental $2$-edge connectivity data structure with poly-logarithmic update time:
\begin{theorem}[A. Kotzig '59~\cite{Kotzig59}]
        Let $G$ be a connected graph with a unique perfect matching $M$. Then $G$ has a bridge that belongs to $M$.
\end{theorem}
\noindent{}The near-linear algorithm for finding a unique perfect matching by Gabow, Kaplan, and Tarjan~\cite{Gabow:2001} is straight-forward: Find a bridge and delete it. If deleting it yields connected components of odd size, it must belong to the matching, and all edges incident to its endpoints may be deleted---if the components have even size, the bridge cannot belong to the matching. Recurse on the components. Thus, to implement Kotzig's Theorem, one has to implement three operations: One that finds a bridge, a second that deletes an edge, and a third returning the size of a connected component. 

Another example is Petersen's theorem~\cite{petersen1891} which states that any cubic, $2$-edge connected graph contains a perfect matching. An algorithm by Biedl et al.~\cite{BIEDL2001110} finds a perfect matching in such graphs in $\Oo(n \log ^4 n)$ time, by using the Holm et al $2$-edge connectivity data structure as a subroutine. In fact, one may implement their algorithm and obtain running time $\Oo(n f(n))$, by using as subroutine a data structure for amortized decremental $2$-edge connectivity with update-time $f(n)$. Here, we thus improve the running time 
from $\Oo(n (\log n)^3 \log \log n)$ 
to $\Oo(n (\log n)^2 (\log \log n)^2)$.

In 2010, Diks and Stanczyk~\cite{Diks2010} improved Biedl et al.'s algorithm for perfect matchings in $2$-edge connected cubic graphs, by having it rely only on dynamic connectivity, not $2$-edge connectivity, and thus obtaining a running time of $\Oo(n (\log n)^2 / \log \log n)$ for the deterministic version, or $\Oo(n \log n (\log \log n)^2)$ expected running time for the randomized version. However, our data structure still yields a direct improvement to the original algorithm by Biedl et al. 

Note that all applications to static graphs have in common that it is no disadvantage that our running time is amortized.

\subsection{Techniques}
As with the previous algorithms, our result is based on top
trees~\cite{Alstrup:2005} which is a hierarchical tree structure used
to represent information about a dynamic tree --- in this case, a
certain spanning tree of the dynamic graph. The original $\Oo((\log
n)^4)$ algorithm of Holm et al.~\cite{Holm:2001} stores $\Oo((\log
n)^2)$ counters with each top tree node, where each counter represent
the size of a certain subgraph. Our new $\Oo((\log n)^3)$ algorithm
applies top trees the same way, representing the same $\Oo((\log
n)^2)$ sizes with each top tree node, but with a much more efficient
implicit representation of the sizes.

Reanalyzing the algorithm of Holm et al.~\cite{Holm:2001}, we show
that many of the sizes represented in the top nodes are identical,
which implies that that they can be represented more efficiently
as a list of actual differences. We then need additional
data structures to provide the desired sizes, and we have to be very
careful when we move information around as the top tree changes,
but overall, we gain almost a log-factor in the amortized time bound,
and the algorithm remains purely combinatorial.

Our combinatorial improvement can be composed with the bittrick
improvement of Thorup~\cite{Thorup:2000}. Thorup represents
the same sizes as the original algorithm of Holm et al., but observes
that we don't need the exact
sizes, but just a constant factor approximation.  Each approximate size
can be represented with only $\Oo(\log\log n)$ bits, and we can
therefore pack $\Omega(\log n/\log\log n)$ of them together in a
single $\Omega(\log n)$-bit word.  This can be used to reduce the cost
of adding two $\Oo(\log n)$-dimensional vectors of approximate sizes
from $\Oo(\log n)$ time to $\Oo(\log\log n)$ time.  It may not be
obvious from the current presentation, but it was a significant
technical difficulty when developing our $\Oo((\log n)^3\log\log n)$
algorithm to make sure we could apply this technique and get the
associated speedup to
$\Oo((\log n)^2(\log\log n)^2)$.

The ``natural'' query time of our algorithm is the same as its update time.
In order to reduce the query time, we observe that we can augment the main algorithm to maintain a secondary structure that can answer queries much faster.  This can be used to reduce the query time for the combinatorial algorithm to $\Oo(\log n)$%
, and for the full algorithm to the optimal $\Oo(\log n/\log\log n)$.

The secondary structure needed for the optimal $\Oo(\log n/\log\log n)$ query time uses top trees of degree $\Oo(\log n/\log\log n)$.  While the use of non-binary trees is nothing new, we believe we are the first to show that such top trees can be maintained in the ``natural'' time.

Finally, we show a general technique for getting down to linear space, using top trees whose base clusters have size $\Theta(\log^c n)$.

\subsection{Article outline}
In Section~\ref{sec:DynTreeOp}, we recall how~\cite{Holm:2001} fundamentally solves $2$-edge connectivity via a reduction to a certain set of operations on a dynamic forest. In Section~\ref{sec:toptrees}, we recall how top trees can be used to maintain information in a dynamic forest, as shown in~\cite{Alstrup:2005}. 
In Sections~\ref{sec:coverlevel},~\ref{sec:findsize}, and~\ref{sec:findfirstlabel}, we describe how to support the operations on a dynamic tree needed to make a combinatorial $\Oo((\log n)^3\log\log n)$ algorithm for bridge finding, as stated in Theorem~\ref{thm:Combinatorial}. 
Then, in Section~\ref{sec:ApproxCount}, we show how to use Approximate Counting to get down to $\Oo((\log n)^2 (\log \log n)^2)$ update time, thus, reaching the update time of Theorem~\ref{thm:MainThm}. 
We then revisit top trees in Section~\ref{sec:toprev}, and introduce the notion of $B$-ary top trees, as well as a general trick to save space in complex top tree applications.
We proceed to show how to obtain the optimal $\Theta(\log n / \log \log n)$ query time in Section~\ref{sec:fastQuery}.
Finally, in Section~\ref{sec:fat}, we show how to achieve optimal space, by only storing cluster information with large clusters, and otherwise calculating it from scratch when needed. 

\section{Reduction to operations on dynamic trees}\label{sec:DynTreeOp}

In~\cite{Holm:2001}, $2$-edge connectivity was maintained via
operations on dynamic trees, as follows.  For each edge $e$ of the
graph, the algorithm explicitly maintains a \emph{level}, $\ell(e)$,
between $0$ and $\lmax=\floor{\log_2 n}$ such that the
edges at level $\lmax$ form a spanning forest $T$, and such that the
$2$-edge connected components in the subgraph induced by edges at
level at least $i$ have at most $\floor{n/2^i}$ vertices. For each
edge $e$ in the spanning forest, define the \emph{cover level},
$c(e)$, as the maximum level of an edge crossing the cut defined by
removing $e$ from $T$, or $-1$ if no such edge exists. The cover levels are only maintained implicitly, because each edge insertion and deletion can change the cover levels of $\Omega(n)$ edges.  Note that the
bridges are exactly the edges in the spanning forest with cover level
$-1$.  The algorithm explicitly maintains the spanning forest $T$ using a
dynamic tree structure supporting the following operations:
\begin{enumerate}[itemsep=-1pt,topsep=2pt]
\item\label{it:link} Link$(v,w)$. Add the edge $(v,w)$ to the dynamic tree, implicitly setting its cover level to $-1$.
\item\label{it:cut} Cut$(v,w)$. Remove the edge $(v,w)$ from the dynamic tree.
\item\label{it:conn} Connected$(v,w)$. Returns $\true$ if $v$ and $w$ are in the same tree, $\false$ otherwise.
\item\label{it:cover} Cover$(v,w,i)$. For each edge $e$ on the tree path from $v$ to $w$ whose cover level is less than $i$, implicitly set the cover level to $i$. (Called when an edge was inserted at level $i=0$ or had its level incresed to $i>0$.)
\item\label{it:uncover} Uncover$(v,w,i)$. For each edge $e$ on the tree path from $v$ to
  $w$ whose cover level is at most $i$, implicitly set the cover level
  to $-1$. (Called when the knowledge we had about whether the edges on the tree
  path from $v$ to $w$ were covered at level $\leq i$ is no longer
  valid because some edge was deleted. This may temporarily set some
  cover levels too low, but the algorithm fixes that using subsequent
  calls to Cover.)
\item\label{it:cl1} CoverLevel$(v)$. Return the minimal cover level of any edge in the tree containing $v$.
\item\label{it:cl2} CoverLevel$(v,w)$. Return the minimal cover level of an edge on the path from $v$ to $w$.  If $v=w$, we define CoverLevel$(v,w)=\lmax$.
\item\label{it:mce1} MinCoveredEdge$(v)$. Return any edge in the tree containing $v$ with minimal cover level. (Find a bridge, anywhere in the tree.)
\item\label{it:mce2} MinCoveredEdge$(v,w)$. Returns a tree-edge on the path from $v$ to $w$ whose cover level is CoverLevel$(v,w)$. (Find a bridge on the given path.)
  \label{it:mce}
\item\label{it:al} AddLabel$(v,l,i)$. Associate the \emph{user label} $l$ to the vertex $v$ at level $i$. (Insert a non-tree edge.)
\item\label{it:rl} RemoveLabel$(l)$.  Remove the user label $l$ from its vertex $\vertex(l)$. (Delete a non-tree edge.)
\item\label{it:ffl} FindFirstLabel$(v,w,i)$.
{\sloppy
  Find a user label at level $i$ such that the associated vertex $u$ has\footnote{$\meet(u,v,w)$ is defined as the unique vertex that is on all simple paths beween any two of $u$, $v$, and $w$.} \mbox{CoverLevel$(u,\meet(u,v,w))\geq i$}; among such user labels, return the one that minimizes the distance from $v$ to $\meet(u,v,w)$. (Find the first candidate witness that part of the path $v \cdots w$ is covered on level $i$, or a non-tree edge to swap with when deleting a covered tree edge.)
}
\item\label{it:fs} FindSize$(v,w,i)$.
  Find the number of vertices $u$ such that \mbox{CoverLevel$(u,\meet(u,v,w))\geq i$}. (Determine the size of the $2$-edge connected component at level $i$ that would result from increasing the level of $(v,w)$ to $i$, or for $v=w$ just find the size of the $2$-edge connected component at level $i$ that contains $v=w$. Thus the size of the $2$-edge component of $v$ in the whole graph is $\FindSize(v,v,0)$.).
  Note that FindSize$(v,v,-1)$ is just the number of vertices in the tree containing $v$ (which is also the size of the connected component of $v$).
\end{enumerate}

\begin{restatable}[Essentially the high level algorithm from~\cite{Holm:2001}]{lemma}{lemhighlevel}
  \label{lem:highlevel}
  There exists a deterministic reduction for dynamic graphs with $n$ nodes, that, when starting with an empty graph, supports any sequence of $m$ Insert or Delete operations using:
  \begin{itemize}
  \item $\Oo(m)$ calls to Link, Cut, Uncover, and CoverLevel.
  \item $\Oo(m\log n)$ calls to Connected, Cover, AddLabel, RemoveLabel, FindFirstLabel, and FindSize.
  \end{itemize}
  And that can answer FindBridge queries using a constant number of calls to Connected, CoverLevel, and MinCoveredEdge, and size queries using a single call to FindSize.
\end{restatable}
\begin{proof}
  See Appendix~\ref{app:highlevel} for a proof and pseudocode.
\end{proof}

\begin{table}[htb!]
    \center
    \begin{adjustbox}{
            max width={\dimexpr(\paperwidth+2\textwidth)/3\relax},
            Trim={\dimexpr(\paperwidth-\textwidth)/6\relax} {0pt}
        }
        \small
        \begin{tabular}{|r|l||c|c|c|c|c|}
            \hline
            \multirow{2}{*}{\#}&\multirow{2}{*}{Operation} & \multicolumn{5}{c|}{Asymptotic worst case time per call, using structure in section}
            \\
            && \ref{sec:coverlevel} & \ref{sec:findsize} & \ref{sec:findfirstlabel} & \ref{sec:ApproxCount} & \ref{sec:fastQuery}
            \\\hline
            \ref{it:link}&Link$(v,w,e)$ & \multirow{9}{*}{$\log n$} & \multirow{9}{*}{$(\log n)^2\log\log n$} & \multirow{9}{*}{$\log n\log\log n$} & \multirow{9}{*}{$\log n(\log\log n)^2$} & \multirow{2}{*}{$\frac{f(n)\log n}{\log f(n)}$}
            \\
            \ref{it:cut}&Cut$(e)$ & & & & & \\\cline{1-2}\cline{7-7}
            \ref{it:conn}&Connected$(v,w)$ & & & & & \multirow{7}{*}{$\frac{\log n}{\log f(n)}$} \\
            \ref{it:cover}&Cover$(v,w,i)$ & & & & & \\
            \ref{it:uncover}&Uncover$(v,w,i)$ & & & & & \\
            \ref{it:cl1}&CoverLevel$(v)$ & & & & & \\
            \ref{it:cl2}&CoverLevel$(v,w)$ & & & & & \\
            \ref{it:mce1}&MinCoveredEdge$(v)$ & & & & & \\
            \ref{it:mce2}&MinCoveredEdge$(v,w)$ & & & & & \\\hline
            \ref{it:al}&AddLabel$(v,l,i)$ & \multirow{3}{*}{-} & \multirow{3}{*}{-} & \multirow{3}{*}{$\log n\log\log n$} & \multirow{3}{*}{-} & \multirow{3}{*}{-}\\
            \ref{it:rl}&RemoveLabel$(l)$ & & & & & \\
            \ref{it:ffl}&FindFirstLabel$(v,w,i)$ & & & & & \\\hline
            \multirow{2}{*}{\ref{it:fs}}
            &FindSize$(v,w,i)$ & - & $(\log n)^2\log\log n$ & - & $\log n(\log\log n)^2$ & - \\\cline{3-7}
            &FindSize$(v,v,-1)$ & $\log n$ & $\log n$ & $\log n$ & $\log n$ & $\frac{\log n}{\log f(n)}$ \\\hline

            \multicolumn{2}{|l||}{} & \multicolumn{5}{c|}{Space cost, using structure in section} \\
            \hline
            \multicolumn{2}{|l||}{natively} & \multirow{2}{*}{$n$} & $n\log n$ & \multirow{2}{*}{$m+n$} & $n\log\log n$ & \multirow{2}{*}{$n$} \\\cline{4-4}\cline{6-6}
            \multicolumn{2}{|l||}{when modified as in Section~\ref{sec:fat}} & & $n$ & & $n$ &
            \\\hline
        \end{tabular}
    \end{adjustbox}
    \caption{\label{tbl:optimes}Data structures presented in this paper. In the last column, $f(n)\in\Oo(\frac{\log n}{\log\log n})$ can be chosen arbitrarily. %
    }
\end{table}

The algorithm in~\cite{Holm:2001} used a dynamic tree structure supporting all the operations in $\Oo((\log n)^3)$ time, leading to an $\Oo((\log n)^4)$ algorithm for bridge finding.  Thorup~\cite{Thorup:2000} showed how to improve the time for the dynamic tree structure to~$\Oo((\log n)^2\log\log n)$ leading to an $\Oo((\log n)^3\log\log n)$ algorithm for bridge finding.

Throughout this paper, we will show a number of data structures for dynamic trees, implementing various subsets of these operations while ignoring the rest (See Table~\ref{tbl:optimes}).  Define a \emph{CoverLevel} structure to be one that implements operations~\ref{it:link}--\ref{it:mce}, and a \emph{FindSize} structure to be a CoverLevel structure that additionally implements the FindSize operation. Finally, we define a \emph{FindFirstLabel} structure to be one that implements operations~\ref{it:link}--\ref{it:ffl} (all except for FindSize).

The point is that we can get different trade-offs between the operation costs in the different structures, and that we can combine them into a single structure supporting all the operations using the following  %

\begin{lemma}[Folklore]\label{lem:combine}
  Given two %
  data structures $S$ and $S'$ for the same problem consisting of a set $U$ of update operations and a set $Q$ of query operations.  If the respective update times are $f_u(n)$ and $f'_u(n)$ for $u\in U$, and the query times are $g_q(n)$ and $g'_q(n)$ for $q\in Q$, we can create a combined data structure running in $\Oo(f_u(n)+f'_u(n))$ time for update operation $u\in U$, and $\Oo(\min\!\set{g_q(n),g'_q(n)})$ time for query operation $q\in Q$.
\end{lemma}
\begin{proof}
  Simply maintain both structures in parallel.  Call all update operations on both structures, and call only the fastest structure for each query.
\end{proof}

\begin{proof}[Proof of Theorem~\ref{thm:Combinatorial}]
  Use the CoverLevel structure from Section~\ref{sec:coverlevel}, the FindSize structure from Section~\ref{sec:findsize}, and the FindFirstLabel structure from Section~\ref{sec:findfirstlabel}, and combine them into a single structure using Lemma~\ref{lem:combine}.  Then the reduction from Lemma~\ref{lem:highlevel} gives the correct running times but uses $\Oo(m+n\log n)$ space.  To get linear space, modify the FindSize and FindFirstLabel structures as described in Section~\ref{sec:fat}.
\end{proof}
\begin{proof}[Proof of Theorem~\ref{thm:MainThm}]
  Use the CoverLevel structure from Section~\ref{sec:fastQuery}, the FindSize structure from Section~\ref{sec:findsize}, as modified in Section~\ref{sec:ApproxCount} and~\ref{sec:fat}, and the FindFirstLabel structure from Section~\ref{sec:findfirstlabel}, and combine them into a single structure using Lemma~\ref{lem:combine}.  Then the reduction from Lemma~\ref{lem:highlevel} gives the required bounds.
\end{proof}

\section{Top trees}\label{sec:toptrees}

A \emph{top tree} is a data structure for maintaining information about each tree of a dynamic forest.
Let $T$ be a tree, and let $\boundary T$ be an arbitrary set of $1$ or $2$ vertices of $T$, which we will call the \emph{external boundary vertices} of $T$.  For any subgraph $S$ of $T$, define the \emph{boundary vertices} of $S$ (denoted $\boundary_{(T,\boundary T)} S$ or just $\boundary S$) as the set of vertices in $S$ that are either in $\boundary T$ or are incident to an edge not in $S$.
A cluster $C$ is a connected subgraph of $T$ with $1$ or $2$ boundary vertices\footnote{Note that this deviates from the existing literature, which introduces a special class of cluster with $0$ boundary vertices, which can only be present in the root~\cite{Alstrup:2005}}.
A top tree $\mathcal{T}$ is a rooted tree representing a recursive partition of $T$ into clusters.  The root of $\mathcal{T}$ corresponds to all of $T$, and each non-leaf node is an edge-disjoint union of the clusters of its children.  The leaves of $\mathcal{T}$ are called \emph{base clusters} and (usually\footnote{We will look at generalized top trees where this is not the case in Section~\ref{sec:toprev}}) correspond to the edges of $T$.

For every cluster $C$ %
 the
\emph{cluster path} of $C$, denoted $\pi(C)$, is the tree path in $T$
connecting $\boundary C$.  If $\sizeof{\boundary C}=2$ then $\pi(C)$
contains at least one edge, and we call $C$ a \emph{path cluster}.
Otherwise $\sizeof{\pi(C)}= 1$ and we call $C$ a \emph{point
  cluster}.  If $\sizeof{\boundary C}=1$ then $\pi(C)$ is the trivial
path consisting of the single boundary vertex.

A top tree is \emph{binary} if each node has at most two children.
We call a non-leaf node \emph{heterogeneous} if it has both a point cluster and a path cluster among its children, and \emph{homogeneous} otherwise.

A path cluster $D$ is called a \emph{path child} of its parent $C$ if $\pi(D)\subseteq\pi(C)$. Note that for binary top trees, a path cluster $D$ is a path child if and only if its parent $C$ is also a path cluster. But for non-binary top trees, even if $C$ and $D$ are both path clusters, $\partial D$ may intersect $\partial C$ only in a point, or not at all.

The top forest supports dynamic changes to the forest: insertion (link) or deletion (cut) of edges. Furthermore, it supports the \emph{expose} operation: expose($v$), or expose($v_1,v_2$), returns a top tree where $v$, or $v_1,v_2$, are external boundary vertices.  %
All operations are supported by performing a series of \emph{destroy}, \emph{create}, \emph{split}, and \emph{merge} operations: \emph{split} destroys a node of the top tree and replaces it with its children, while merge creates a parent as a union of its children. Destroy and create are the base cases for split and merge, respectively. Note that clusters can only be merged if they are edge-disjoint and their union is a cluster (i.e. is connected and has a boundary of size at most $2$).

\begin{theorem}[Alstrup, Holm, de Lichtenberg, Thorup~\cite{Alstrup:2005}]
For a dynamic forest on $n$ vertices we can maintain binary top trees of height $\Oo(\log n)$ supporting each link, cut or expose with a sequence of $\Oo(1)$ calls to create or destroy, and $\Oo(\log n)$ calls to merge or split. These top tree modifications are identified in $\Oo(\log n)$ time. The space usage of the top trees is linear in the
size of the dynamic forest.
\end{theorem}

\section{A $\CoverLevel$ structure}\label{sec:coverlevel}
In this section we show how to maintain a top tree supporting the
CoverLevel operations.  This part is essentially the same as
in~\cite{Holm:1998,Holm:2001} (with minor corrections), but is included here for completeness because the rest of the paper builds on it.  Pseudocode for maintaining this structure is given in Appendix~\ref{app:coverlevel}.

For each cluster $C$ we want to maintain the following two integers
and up to two edges:
\begin{align*}
  \cover_C &:= \min\set{c(e)\cond e\in\pi(C)}\cup\set{\lmax}
  \\
  \globalcover_C &:= \min\set{c(e)\cond e\in C\setminus\pi(C)}\cup\set{\lmax}
  \\
  \minpathedge_C &:=
    \mathrlap{\argmin\limits_{e\in\pi(C)} c(e)}\hphantom{\argmin\limits_{e\in C\setminus\pi(C)} c(e)}
    \text{ if $\abs{\boundary C}=2$, and }
    \nil \text{ otherwise}
  \\
  \minglobaledge_C &:=
    \argmin\limits_{e\in C\setminus\pi(C)} c(e) \text{ if $C\neq\pi (C)$, and }
    \nil \text{ otherwise}
\end{align*}
Then
\begin{align*}
  \left.
  \begin{aligned}
    \CoverLevel(v) &= \globalcover_C
    \\
    \MinCoveredEdge(v) &= \minglobaledge_C
  \end{aligned}
  \right\}
  & \text{ where $C$ is the point cluster returned by $\Expose(v)$}
  \\
  \left.
  \begin{aligned}
    \CoverLevel(v,w) &= \cover_C
    \\
    \MinCoveredEdge(v,w) &= \minpathedge_C
  \end{aligned}
  \right\}
  &\text{ where $C$ is the path cluster returned by $\Expose(v,w)$}
\end{align*}

The problem is that when handling Cover or Uncover we cannot afford to
propagate the information all the way down to the edges.  When these
operations are called on a path cluster $C$, we instead implement them
directly in $C$, and then store ``lazy information'' in $C$ about what
should be propagated down in case we want to look at the descendants
of $C$.  The exact additional information we store for a path cluster $C$ is
\begin{align*}
  \cover^-_C &:= \text{max level of a pending Uncover, or $-1$}
  \\
  \cover^+_C &:= \text{max level of a pending Cover, or $-1$}
\end{align*}
We maintain the invariant that $\cover_C\geq\cover^+_C$, and if $\cover_C\leq\cover^-_C$ then $\cover_C=\cover^+_C$.

This allows us to implement Cover$(v,w,i)$ by first calling $\Expose(v,w)$, and then updating the returned path cluster $C$ as follows:
\begin{align*}
  \cover_C &= \max\set{\cover_C, i}
  &
  \cover^+_C &= \max\set{\cover^+_C, i}
\end{align*}
Similarly, we can implement Uncover$(v,w,i)$ by first calling $\Expose(v,w)$, and then updating the returned path cluster $C$ as follows if $\cover_C\leq i$:
\begin{align*}
  \cover_C &= -1
  &
  \cover^+_C &= -1
  &
  \cover^-_C &= \max\set{\cover^-_C, i}
\end{align*}

Together, $\cover^-_C$ and $\cover^+_C$ represent the fact that for each path descendant $D$ of $C$, if $\cover_D\leq\max\set{\cover^-_C,\cover^+_C}$\footnote{In~\cite{Holm:1998,Holm:2001} this condition is erroneously stated as $\cover_D\leq\cover^-_C$.}, we need to set $\cover_D=\cover^+_C$.
In particular whenever a path cluster $C$ is split, for each path child $D$ of $C$, if $\max\set{\cover_D,\cover^-_D}\leq\cover^-_C$ we need to set
\begin{align*}
  \cover^-_D &= \cover^-_C
\intertext{Furthermore, if $\cover_D\leq\max\set{\cover^-_C,\cover^+_C}$ we need to set}
  \cover_D &= \cover^+_C
  &
  \cover^+_D &= \cover^+_C
\end{align*}
Note that only $\cover_D$ is affected. None of $\globalcover_D$, $\minpathedge_D$, or $\minglobaledge_D$ depend directly on the lazy information.

Now suppose we have $k$ clusters\footnote{$k=2$ for now, but we will reuse this in section~\ref{sec:fastQuery} with a higher-degree top tree.} $A_1,\ldots,A_k$ that we want to merge into a single new cluster $C$. For $1\leq i\leq k$ define
\begin{align*}
  \globalcover'_{C,A_i} &:=
  \begin{cases}
    \mathrlap{\globalcover_{A_i}}\hphantom{\minglobaledge_{A_i}}
    &\text{if }\boundary A_i\subseteq \pi(C)
    \text{ or }\globalcover_{A_i}\leq\cover_{A_i}
    \\
    \cover_{A_i}
    &\text{otherwise}
  \end{cases}
  \\
  \minglobaledge'_{C,A_i} &:=
  \begin{cases}
    \minglobaledge_{A_i}
      &\text{if }\boundary A_i\subseteq \pi(C)
    \text{ or }\globalcover_{A_i}\leq\cover_{A_i}
    \\
    \minpathedge_{A_i}
      &\text{otherwise}
  \end{cases}
\end{align*}
Note that for a point-cluster $A_i$, $\globalcover_{A_i}$ is always $\leq \cover_{A_i}=l_{\max}$. 

We then have the following relations between the data of the parent and the data of its children:
\begin{align*}
  \cover_C &= \lmax\text{ if }\abs{\boundary C}<2,\text{ otherwise }
  \min\limits_{%
    1\leq i<k,
    \boundary A_i\subseteq \pi(C)
  } \cover_{A_i}
  \\
  \minpathedge_C &= \mathrlap{\nil}\hphantom{\lmax}\text{ if }\abs{\boundary C}<2,\text{ otherwise }
  \minpathedge_{A_j}
    \text{ where }j=\argmin\limits_{%
            1\leq i<k, \boundary A_i\subseteq \pi(C)
                } \cover_{A_i}
  \\
  \globalcover_C &= \min_{1\leq i<k}\globalcover'_{C,A_i}
  \\
  \minglobaledge_C &= \minglobaledge'_{C,A_j}
  \quad\text{where }j=\argmin_{1\leq i<k} \globalcover'_{C,A_i}
  \\
  \cover^-_C &= -1
  \\
  \cover^+_C &= -1
\end{align*}

\paragraph{Analysis}
For any constant-degree top tree, Merge and Split with this information takes constant time, and thus, all operations in the CoverLevel structure in this section take $\Oo(\log n)$ time.  Each cluster uses $\Oo(1)$ space, so the total space used is $\Oo(n)$.

\section{A $\FindSize$ structure}\label{sec:findsize}
We now proceed to show how to extend the CoverLevel structure from Section~\ref{sec:coverlevel} to support FindSize in $\Oo(\log n\log\log n)$ time per Merge and Split.  Later, in Section~\ref{sec:ApproxCount} we will show how to reduce this to $\Oo((\log\log n)^2)$ time per Merge and Split.
See Appendix~\ref{app:findsize} for pseudocode.

We will use the idea of having a single \emph{vertex label} for each vertex, which is a point cluster with no edges, having that vertex as boundary vertex and containing all relevant information about the vertex. The advantage of this is that it simplifies handling of the common boundary vertex during a merge by making sure it is uniquely assigned to (and accounted for by) one of the children.

Let $C$ be a cluster in $T$, let $v$ be a 
vertex in $\pi(C)$, and let $0\leq i<\lmax$. Define
\begin{align*}
  \pointset_{C,v,i} &:= \set{ u\in C \cond
      \begin{gathered}
        \pi(C)\cap u\cdots v=\set{v} \\
        \:\wedge\: \CoverLevel(u,v)\geq i
      \end{gathered}
    }
\intertext{Intuitively, $\pointset_{C,v,i}$ is the set of vertices in $C$ whose path to $v$ is covered at level $\geq i$ independently of the cover levels on $\pi(C)$. Information (such as the size or the existence of certain marked vertices) about this set stays constant for as long as $C$ exists, no matter what happens with the lazy information in the ancestors to $C$.  In this section we only care about the size}
  \pointsize_{C,v,i} &:= \sizeof{\pointset_{C,v,i}}
\intertext{For convenience, we will combine all the $\Oo(\log n)$ levels together into a single vector\footnotemark}
  \pointsize_{C,v} &:= \left( \pointsize_{C,v,i} \right)_{\set{0\leq i<\lmax}}
\end{align*}%
\footnotetext{All vectors and matrices in this section have indices ranging from $0$ to $\lmax-1$.}%
Then we can define the vector
\begin{align*}
  \size_{C}
  &:=
  \sum_{u\in\pi(C)}\pointsize_{C,u}
\end{align*}
Note that with this definition, if $\boundary C=\set{v}$ then $\pointsize_{C,v}=\size_C$  so even when $v=w$ we have
\begin{align*}
  \FindSize(v,w,i) = \size_{C,i}\qquad\text{where }C = \Expose(v,w)
\end{align*}
So for any cluster $C$, the $\size_C$ vector is what we want to maintain.

The main difficulty turns out be computing the $\size_{C}$ vector for the heterogeneous point clusters.  To help with that we will for each cluster $C$ and boundary vertex $v\in\boundary C$ {break $\pi(C)$ into $\lmax+2$ \emph{parts}.  For each $-1\leq i\leq\lmax$ define
\begin{align*}
  \partpath_{C,v,i} &:= \set{u\in\pi(C)\cond\CoverLevel(u,v)=i}
\end{align*}
Then $\partpath_{C,v,i}$ (if nonempty) is a contiguous subset of the vertices on $\pi(C)$.  Furthermore, $\partpath_{C,v,\lmax}=\set{v}$,  $\boundary C\setminus\set{v}\subseteq\partpath_{C,v,-1}$, and for all $0\leq i<\lmax$ the set $\partpath_{C,v,i}$ lies between the closest edge $e$ to $v$ with $c(e)\leq i$ and the closest edge $e'$ to $v$ with $c(e')<i$.  In addition to the $\size_C$ vector, we will
}%
maintain the following two size vectors for each part:
\begin{align*}
  \partsize_{C,v,i}
  &:=
  \sum_{\mathclap{u\in\partpath_{\mathrlap{C,v,i}}}} \pointsize_{C,u}
  &
  \diagsize_{C,v,i}
  &:=
  M\!(i)\cdot\partsize_{C,v,i}
\end{align*}
Where $M\!(i)$ is a diagonal matrix whose entries are defined 
by\footnote{Here, \mbox{$[P]=\begin{cases}1&\text{if }P\text{ is true}\\0&\text{otherwise}\end{cases}$} is the \emph{Iverson Bracket} (see~\cite{10.2307/2325085}).}
\begin{align*}
  M\!(i)_{jj} &= [j\leq i]
\end{align*}
The $M\!(i)$ matrix is purely a notational convenience whose purpose is to ``zero out'' some elements in a vector. In particular, for $0\leq j<\lmax$
\begin{align*}
  \diagsize_{C,v,i,j} =
  (M\!(i)\cdot \partsize_{C,v,i})_j =
    \begin{cases}
      \partsize_{C,v,i,j}&\text{if $j\leq i$}
      \\
      0&\text{otherwise}
    \end{cases}
\end{align*}
Note that these vectors do not take $\cover^-_C$ and $\cover^+_C$ (as defined in Section~\ref{sec:coverlevel}) into account.  The corresponding  ``clean'' vectors are not explicitly stored, but computed when needed as follows
\begin{align*}
  \left.
  \begin{aligned}
    \partsize'_{C,v,i}
    &=
    \begin{cases}
      \partsize_{C,v,i}
      &\text{if $i>\ell$}
      \\
      \mathrlap{
        \sum_{j=-1}^{\ell}\partsize_{C,v,j}
      }
      \hphantom{M\!(i)\cdot\sum_{j=-1}^{\ell}\partsize_{C,v,j}}
      &\text{if $i=\cover^+_C$}
      \\
      \vec{0}
      &\text{otherwise}
    \end{cases}
    \\
    \diagsize'_{C,v,i}
    &=
    \begin{cases}
      \diagsize_{C,v,i}
      &\text{if $i>\ell$}
      \\
      M\!(i)\cdot\sum_{j=-1}^{\ell}\partsize_{C,v,j}
      &\text{if $i=\cover^+_C$}
      \\
      \vec{0}
      &\text{otherwise}
    \end{cases}
  \end{aligned}
  \right\}
  \text{where $\ell=\max\set{\cover^-_C,\cover^+_C}$}
\end{align*}

The point of these definitions is that each path cluster inherits most
of its $\partsize$ and $\diagsize$ vectors from its children, and we can use this fact to get an $\Oo(\lmax/\log\lmax)=\Oo(\log n/\log\log n)$ speedup compared to~\cite{Holm:2001}.

\paragraph{Merging along a path (the general case)}
Let $A,B$ be clusters that we want to merge into a new cluster $C$, and suppose $\boundary A\cup\boundary B\subseteq\pi(C)$.  This covers both types of homogeneous merges (two point or two path clusters), as well as the heterogeneous merge (one point and one path cluster) where the result is a path cluster.  The only type of merge not covered is the heterogeneous merge resulting in a point cluster, which is handled in the next section.
Let $\boundary A\cap\boundary B=\set{c}$.  If $\abs{\boundary C}=1$, let $a=b=c$, otherwise let $\boundary C=\set{a,b}$ with $a\in\boundary A$, $b\in\boundary B$. Then
\begin{align*}
  \size_{C} &= \size_{A} + \size_{B}
  \\
  \partsize_{C,a,i}
  &=
  \begin{cases}
    \partsize'_{A,a,i}
    &\text{if $i>\cover_{A}$}
    \\
    \mathrlap{
      \partsize'_{A,a,i} + \sum_{j=i}^{\lmax}\partsize'_{B,c,j}
    }
    \hphantom{\diagsize'_{A,a,i} + M\!(i)\cdot\sum_{j=i}^{\lmax}\partsize'_{B,c,j}}
    &\text{if $i=\cover_{A}$}
    \\
    \partsize'_{B,c,i}
    &\text{if $i<\cover_{A}$}
  \end{cases}
  \\
  \diagsize_{C,a,i}
  &=
  \begin{cases}
    \diagsize'_{A,a,i}
    &\text{if $i>\cover_{A}$}
    \\
    \diagsize'_{A,a,i} + M\!(i)\cdot\sum_{j=i}^{\lmax}\partsize'_{B,c,j}
    &\text{if $i=\cover_{A}$}
    \\
    \diagsize'_{B,c,i}
    &\text{if $i<\cover_{A}$}
  \end{cases}
\end{align*}
The formulas for $\partsize_{C,b,i}$ and $\diagsize_{C,b,i}$ are
analoguous. The important thing to note is that if we have already computed and stored the $\partsize'_{A,a,i}$, $\partsize'_{B,c,i}$, $\diagsize'_{A,a,i}$, and $\diagsize'_{B,c,i}$ vectors for all $i$, then the only new value we need to compute is 
for $i=\cover_A$.
The rest can be inherited.

\paragraph{Merging off the path (heterogeneous point clusters)}
Now let $A$ be a path cluster with $\boundary A=\set{a,b}$, let $B$ be a point cluster with $\boundary B=\set{b}$, and suppose we want to merge $A,B$ into a new point cluster $C$ with $\boundary C=\set{a}$. Then
\begin{align*}
  \size_{C} &= \left(\sum_{i=-1}^{\lmax}\diagsize'_{A,a,i}\right) + M\!(\cover_A)\cdot \size_{B}
  \\
  \partsize_{C,a,i} &=
  \begin{cases}
    \size_{C}&\text{if $i=\lmax$}
    \\
    \vec{0}&\text{otherwise}
  \end{cases}
  \\
  \diagsize_{C,a,i} &= \partsize_{C,a,i}
\end{align*}

\paragraph{Analysis}
The advantage of our new approach is that each merge or split is a \emph{constant} number of splits, concatenations, searches, and sums over $\Oo(\lmax)$-length lists of $\lmax$-dimensional vectors.  
By representing each list as an augmented balanced binary search tree (see e.g.~\cite[pp. 471--475]{Knuth:1998}), we can implement each of these operations in $\Oo(\lmax\log\lmax)$ time, and using $\Oo(\lmax)$ space per cluster, as follows.  Let $C$ be a cluster and let $v\in\boundary C$.  The tree has one node for each key $i, -1\leq i\leq\lmax$ such that $\partsize_{C,v,i}$ is nonzero, augmented with the following additional information:
\begin{align*}
  \operatorname{key} &:= i
  \\
  \partsize &:= \partsize_{C,v,i}
  \\
  \diagsize &:= \diagsize_{C,v,i}
  \\
  \partsizesum &:= 
  \sum_{j\text{ descendant of }i}\partsize_{C,v,j}
  \\
  \diagsizesum &:= 
  \sum_{j\text{ descendant of }i}\partsize_{C,v,j}
\end{align*}
Each split, concatenate, search, or sum operation can be implemented such that it touches $\Oo(\log\lmax)$ nodes, and the time for each node update is dominated by the time it takes to add two $\lmax$-dimensional vectors, which is $\Oo(\lmax)$.
The total time for each Cover, Uncover, Link, Cut, or FindSize is therefore $\Oo(\log n\cdot\lmax\cdot\log\lmax)=\Oo((\log n)^2\log\log n)$, and the total space used for the structure is $\Oo(n\cdot\lmax)=\Oo(n\log n)$.

\paragraph{Comparison to previous algorithms}
For any path cluster $C$ and vertex $v\in \boundary C$, let $S_{C,v}$ be the matrix whose $j$th column $0\leq j<\lmax$ is defined by
\begin{align*}
  (S^T_{C,v})_j &:= \sum_{k=j}^{\lmax}\partsize'_{C,v,k}
\end{align*}
Then $S_{C,v}$ is essentially the $\size$ matrix maintained for path clusters in~\cite{Holm:1998,Thorup:2000,Holm:2001}.  Notice that
\begin{align*}
  \diag(S_{C,v}) &= \sum_{k=-1}^{\lmax}\diagsize'_{C,v,k}
\end{align*}
which explains our choice of the ``diag'' prefix.

\section{A $\FindFirstLabel$ structure}\label{sec:findfirstlabel}
We will show how to maintain information that allows us to implement  FindFirstLabel; the function that allows us to inspect the replacement edge candidates at a given level. The implementation uses a ``destructive binary search, with undo'' strategy, similar to the non-local search introduced in~\cite{Alstrup:2005}.

The idea is to maintain enough information in each cluster to determine if there is a result.  Then we can start by using $\Expose(v,w)$, and repeatedly split the root containing the answer until we arrive at the correct label.  After that, we simply undo the splits (using the appropriate merges), and finally undo the $\Expose$.

Just as in the FindSize structure, we will use vertex labels to store all the information pertinent to a vertex.
We store all the added \emph{user labels} for each vertex in the label object for that vertex in the base level of the top tree.  For each level where the vertex has an associated user label, we keep a doubly linked list of those labels, and we keep a singly-linked list of these nonempty lists.  Thus, $\FindFirstLabel(v,w,i)$ boils down to finding the first vertex label that has an associated user label at the right level.  Once we have that vertex label, the desired user label can be found in $\Oo(\lmax)$ time.

Let $C$ be a cluster in $T$, and let $v\in\boundary C$%
. Define bit vectors\footnote{Again, using the Iverson bracket.}
\begin{align*}
  \pointincident_{C,v} &:=
  \left(\left[
    \exists v\in \pointset_{C,v,i}\!:\parbox{5em}{$v$ has labels at level $i$}
  \right]\right)_{\!\!\set{0\leq i<\lmax}}
  \\
  \incident_{C} &:=
\bigvee_{u\in\pi(C)}\pointincident_{C,u}
\end{align*}

Maintaining the $\incident_{C}$ bit vectors, and the corresponding $\partincident_{C,v}$ and $\diagincident_{C,v}$ bit vectors, can be done completely analogous to the way we maintain the $\size$ vectors used for FindSize, with the minor change that we use bitwise OR on bit vectors instead of vector addition.

Updating the vertex label cluster $C$ in the top tree during $\AddLabel(v,l,i)$, or a $\RemoveLabel(l)$ where $v=\vertex(l)$ and $\ell(l)=i$ can be done by first calling detach$(C)$, then updating the linked lists containing the user labels and setting
\begin{align*}
  \incident_{C} &= ([v\text{ has labels at level }j])_{\set{0\leq j<\lmax}}
  \\
  \partincident_{C,v,i} &=
  \begin{cases}
    \incident_{C}&\text{if $i=\lmax$}
    \\
    \vec{0}&\text{otherwise}
  \end{cases}
  \\
  \diagincident_{C,v} &= \partincident_{C}
\end{align*}
and then reattaching $C$.  Finally FindFirstLabel(v,w,i) can be implemented in the way already described, by examining $\pointincident_{C,v,i}$ for each cluster.  Note that even though we don't explicitly maintain it, for any cluster $C$ and any $v\in\boundary C$ we can easily compute
\begin{align*}
  \pointincident_{C,v}
  &=
  \bigvee_{i=-1}^{\lmax}\diagincident'_{C,v,i}
  \\
  &=
  \left(
  \bigvee_{i=\ell+1}^{\lmax}\diagincident_{C,v,i}
  \right)
  + M\!(\cover_C^+)\cdot
  \left(
  \bigvee_{i=-1}^{\ell}\partincident_{C,v,i}
  \right)
  \\
  &\qquad\text{where }\ell:=\max\set{\cover_C^-,\cover_C^+}
\end{align*}

In general, let $A_1,\ldots, A_k$ be the clusters resulting from an expose or split, let $v,w\in\bigcup_{i=1}^k\boundary A_i$ (not necessarily distinct). Then we can define
\begin{align*}
  \FindFirstLabel((A_1,\ldots,A_k); v,w,i)
  &=
  \begin{cases}
    \operatorname{userlabels}_{v_x,i}&\text{if $A_x$ is a vertex label}
    \\
    \FindFirstLabel(\Split(A_x);v_x,w_x,i)&\text{otherwise}
  \end{cases}
  \\
  \text{where }&\text{for }1\leq j\leq k
  \\
  v_j&=\argmin_{u\in\boundary A_j}\dist(v,u)
  \\
  w_j&=\argmax_{u\in\boundary A_j}\dist(v,u)
  \\
  \text{and }&
  \\
  I &= \set{1\leq j \leq k\cond 
    \begin{gathered}
      \CoverLevel(v,v_j)\geq i \\\quad\wedge\quad  \pointincident_{A_j,v_j,i}=1
    \end{gathered}
  }
  \\
  x &= \argmin_{j\in I} \:(3\cdot\dist(v,\meet(v_j,v,w))+\abs{\partial A_j\cap v\cdots w})
  \\
  \FindFirstLabel(v,w,i)
  &=\FindFirstLabel(\Expose(v,w);v,w,i)
\end{align*}
What happens here is that, for each $j$, the vertices $v_j$ and $w_j$ are the boundary vertices of $A_j$ closest to $v$, and farthest from $v$, respectively. Thus, if $A_j$ is a path cluster, $\boundary A_j=\set{v_j,w_j}$, otherwise $\boundary A_j=\set{v_j}=\set{w_j}$. %
The set $I$ defined above is the set of indices of the clusters that contain labels at level $i$. Then, $x$ is picked from $I$ to minimize $D(x)=\dist(v, \meet(v_x,v,w))$.  If there are more than one cluster minimizing $D(x)$, we prefer clusters with at most one boundary vertex on $\pi(C)$, since any vertex $u$ in such a cluster will have $\dist(v, \meet(u,v,w))=D(x)$, which is minimal. A path cluster $A_x$ with both boundary vertices on $\pi(C)$ is only picked if it is the only cluster minimizing $D(x)$.  In either case, we know that $A_x$ contains a vertex $u$ with the desired label, and that any vertex in $A_x$ minimizing $\dist(v_x, \meet(u,v_x,w_x))$ will suffice.  Now if $A_x$ is a vertex label, it has only one vertex, and it stores the desired user label.  Otherwise, we simply split $A_x$ and recurse

\paragraph{Analysis}
By the method described in this section, AddLabel, RemoveLabel, and FindFirstLabel are maintained in $\Oo(\log n\cdot\lmax\cdot\log\lmax)=\Oo((\log n)^2\log\log n)$ worst-case time.

This can be reduced to $\Oo(\log n\cdot\log\lmax)=\Oo(\log n\log\log n)$ by realizing that each $\lmax$-dimensional bit vector fits into $\Oo(1)$ words, and that each bitwise OR therefore only takes constant time.

The total space used for a FindFirstLabel structure with $n$ vertices and $m$ labels is $\Oo(m+n)$ plus the space for $\Oo(n)$ bit vectors. If we assume a word size of $\Omega(\log n)$, this is just $\Oo(m+n)$ in total. If we disallow bit packing tricks, we may have to use $\Oo(m+n\cdot\lmax)=\Oo(m+n\log n)$ space.

\section{Approximate counting}\label{sec:ApproxCount}
As noted in~\cite{Thorup:2000}, we don't need to use the exact
component sizes at each level.  If $s$ is the actual correct size, it
is sufficient to store an approximate value $s'$ such that $s'\leq s\leq e^{\epsilon}s'$, for some constant $0<\epsilon<\ln2$.  Then we are no longer guaranteed that component sizes drop by a factor of $\frac{1}{2}$ at each level, but rather get a factor of $\frac{e^{\epsilon}}{2}$. This increases the number of levels to $\lmax=\floor{\ln n/(\ln2-\epsilon)}$ (which is still $\Oo(\log n)$), but leaves the algorithm otherwise unchanged.  Suppose we represent each size as a floating point value with a $b$-bit mantissa, for some $b$ to be determined later.  For each addition of such numbers the relative error increases.  The relative error at the root of a tree of additions of height $h$ is $(1+2^{-b})^h\leq e^{2^{-b}h}$, thus to get the required precision it is sufficient to set $b=\log_2\frac{h}{\epsilon}$.  In our algorithm(s) the depth of calculation is clearly upper bounded by $h\leq h(n)\cdot\lmax$, where $h(n)=\Oo(\log n)$ is the height of the top tree.  It follows that some $b\in\Oo(\log\log n)$ is sufficient.  Since the maximum size of a component is $n$, the exponent has size at most $\ceil{\log_2n}$, and can be represented in $\ceil{\log_2\ceil{\log_2 n}}$ bits.  Thus storing the sizes as $\Oo(\log\log n)$ bit floating point values is sufficient to get the required precision.
Assuming a word size of $\Omega(\log n)$ this lets us store $\Oo\!\big(\frac{\log n}{\log \log n}\big)$ sizes in a single word, and to add them in parallel in constant time.

\paragraph{Analysis}
We will show how this applies to our FindSize structure from Section~\ref{sec:findsize}.
The bottlenecks in the algorithm all have to do with operations on $\lmax$-dimensional size vectors.  In particular, the amortized update time is dominated by the time to do $\Oo(\log n\cdot\log\lmax)$ vector additions, and $\Oo(\log n)$ multiplications of a vector by the $M\!(i)$ matrix.  With approximate counting, the vector additions each take $\Oo(\log\log n)$ time.  Multiplying a size vector $x$ by $M\!(i)$ we get:
\begin{align*}
  (M\!(i)\cdot x)_j &=
  \begin{cases}
    x_j&\text{if $j\leq i$}
    \\
    0&\text{otherwise}
  \end{cases}
\end{align*}
And clearly this operation can also be done on $\Oo\!\big(\frac{\log n}{\log \log n}\big)$ sizes in parallel when they are packed into a single word.  With approximate counting, each multiplication by $M\!(i)$ therefore also takes $\Oo(\log\log n)$ time.  Thus the time per operation is reduced to $\Oo(\log n(\log\log n)^2)$.

The space consumption of the data structure is $\Oo(n)$ plus the space needed to store $\Oo(n)$ of the $\lmax$-dimensional size vectors.  With approximate counting that drops to $\Oo(\log\log n)$ per vector, or $\Oo(n\log\log n)$ in total.

\paragraph{Comparison to previous algorithms}
Combining the modified FindSize structure with the CoverLevel structure from Section~\ref{sec:coverlevel} and the FindFirstLabel structure from Section~\ref{sec:findfirstlabel} gives us the first bridge-finding structure with $\Oo((\log n)^2(\log\log n)^2)$ amortized update time.  This structure uses $\Oo(m+n\log\log n)$ space, and uses $\Oo(\log n)$ time for FindBridge and Size queries, and $\Oo(\log n(\log\log n)^2)$ for $2$-size queries.

For comparison, applying this trick in the obvious way to the basic $\Oo((\log n)^4)$ time and $\Oo(m+n(\log n)^2)$ space algorithm from~\cite{Holm:1998, Holm:2001} gives the $\Oo((\log n)^3\log n)$ time and $\Oo(m+n\log n\log\log n)$ space algorithm briefly mentioned in~\cite{Thorup:2000}.

\section{Top trees revisited}\label{sec:toprev}
We can combine the tree data structures presented so far to build a data structure for bridge-finding that has update time $\Oo((\log n)^2(\log\log n)^2)$, query time $\Oo(\log n)$, and uses $\Oo(m + n\log\log n)$ space.

In order to get faster queries and linear space, we need to use top-trees in an even smarter way. For this, we need the full generality of the top trees described in~\cite{Alstrup:2005}.

\subsection{Level-based top trees, labels, and fat-bottomed trees}

As described in~\cite{Alstrup:2005}, we may associate a level with each cluster, such that the leaves of the top tree have level $0$, and such that the parent of a level $i$ cluster is on level $i+1$. 
As observed in Alstrup et al.~\cite[Theorem 5.1]{Alstrup:2005}, one may also associate one or more \emph{labels} with each vertex. For any vertex, $v$, we may handle the label(s) of $v$ as point clusters with $v$ as their boundary vertex and no edges.
Furthermore, as described in~\cite{Alstrup:2005}, we need not have single edges on the bottom most level. We may generalize this to instead have clusters of \emph{size} $\le Q$, that is, with at most $Q$ edges, as the leaves of the top tree.

\begin{theorem}[Alstrup, Holm, de Lichtenberg, Thorup~\cite{Alstrup:2005}]\label{thm:toptrees-full}
Consider a fully dynamic forest and let $Q$ be a positive integer
parameter. For the trees in the forest, we can maintain levelled top 
trees whose base clusters are of size at most $Q$ and such that if a tree 
has size s, it has height $h = \Oo(\log s)$ and $\ceil{\Oo(s/(Q(1 + 
    \varepsilon)^i))}$ clusters on level $i \le h$. Here, 
$\varepsilon$ is a positive constant. Each link, 
cut, attach, detach, or expose operation is supported with $\Oo(1)$
creates and destroys, and $\Oo(1)$ joins and splits on each positive level.
If the involved trees have total size $s$, this involves $\Oo(\log s)$ top 
tree modifications, all of which are identified in $\Oo(Q + \log s)$ time.
For a composite sequence of $k$ updates, each of the above bounds are 
multiplied by $k$. As a variant, if we have parameter $S$ bounding the size 
of each underlying tree, then we can choose to let all top roots be on the
same level $H = \Oo(\log S)$.
\end{theorem} 

\subsection{High degree top trees}
Top trees of degree two are well described and often used. However, it turns out to be useful to also consider top trees of higher degree $B$, especially for $B\in\omega(1)$. 
\begin{lemma}\label{lem:highdeg}
    Given any $Q\ge 1$ and $B\ge 2$, one can maintain top trees of degree $B$ and height
    $\Oo(\log n/\log B)$ with base clusters of size at most $Q$. Each expose, link, or cut is handled by $\Oo(1)$ calls 
    to create or destroy and $\Oo(\log n/\log B)$ calls to split or merge. 
    The operations are identified in $\Oo(B(\log n/\log B)+Q)$ time.
\end{lemma} 
\begin{proof}
    Given a binary levelled top tree $\mathcal{T}_2$ of height $h$ with base clusters of size at most $Q$ as in Theorem~\ref{thm:toptrees-full}, we can create a $B$-ary levelled top tree $\mathcal{T}_B$, where the leaves of $\mathcal{T}_B$ are the leaves of $\mathcal{T}_2$, and where the clusters on level $i$ of $\mathcal{T}_B$ are the clusters on level $i\cdot \floor{\log_2 B}$ of $\mathcal{T}_2$. Edges in $\mathcal{T}_B$ correspond to paths of length $\floor{\log_2 B}$ in $\mathcal{T}_2$. Thus, given a binary top tree, we may create a $B$-ary top tree bottom-up in linear time.
 
    We may implement link, cut and expose by running the corresponding operation in $\mathcal{T}_2$. Each cut, link or expose operation will affect clusters on a constant number of root-paths in $\mathcal{T}_2$. There are thus only $\Oo(\log n/\log B)$ calls to split or merge of a cluster on a level divisible by $\floor{\log _2 B}$. Thus, since each split or merge in $\mathcal{T}_B$ corresponds to a split or merge of a cluster in $\mathcal{T}_2$ whose level is divisible by $\floor{\log _2 B}$, we have only $\Oo(\log n /\log B)$ calls to split and merge in $\mathcal{T}_B$.
    
    However, since there are $\Oo(B)$ clusters whose parent pointers need to be updated after a merge, the total running time becomes $\Oo(B(\log n/\log B)+Q)$.
\end{proof}

\subsection{Saving space with fat-bottomed top trees}
In this section we present a general technique for reducing the space usage of a top tree based data structure to linear.  For convenience, we will call any $b(n)$-ary top tree data structure that can be implemented using the top trees from Lemma~\ref{lem:highdeg} \emph{well-behaved}.  Loosely speaking, any well-behaved top tree data structure can be modified to use linear space.

The properties of the technique are captured in the following:
\begin{lemma}\label{lem:savespace}
	Suppose we have a well-behaved $b(n)$-ary top tree data structure, that uses $s(n)$ space per cluster, and spends $t(n)$ worst-case time per merge or split.
	Suppose further that there exists an algorithm that takes any subgraph of size $q$ that forms a cluster, say, $C$, and calculates the complete information for $C$ in $t_0(q,n)$ time, and suppose that the complete information for $C$ has size at most $s_0(q,n)$.
	Finally, suppose that there exists a function $q$ of $n$ such that $s(n)<s_0(q(n),n)\in\Oo(q(n))$.
	
	Then, there exists a data structure maintaining the same information in top trees of height $h=O(\log n /\log b(n))$, such that the top trees use linear space in total, and have $\Oo(t(n)\cdot h(n) + t_0(q(n),n))$ update time for link, cut, and expose.
\end{lemma}

\begin{proof}
  This follows directly from Lemma~\ref{lem:highdeg} by setting $Q=q(n)$ and $B=b(n)$.  Then the top tree will have $\Oo(n / q(n))$ clusters of size at most $s_0(q(n),n)=\Oo(q(n))$ so the total size is linear.  The time per update follows because the top tree uses $\Oo(h(n))$ merges or splits and $\Oo(1)$ create and destroy per link cut and expose.  These take $t(n)$ and $t_0(q(n),n)$ time respectively.
\end{proof}

\section{A faster $\CoverLevel$ structure}\label{sec:fastQuery}
If we allow ourselves to use bit tricks, we can improve the
$\CoverLevel$ data structure from Section~\ref{sec:coverlevel}.
The main idea is, for some $0<\epsilon<1$, to use top trees of degree $b(n)=(\log n)^\epsilon\in\Oo(w/\log\lmax)$. As noted in Lemma~\ref{lem:highdeg}, such
top trees have height $h(n)\in\Oo(\frac{\log n}{\epsilon\log\log n})$, and finding the sequence of merges and splits for a given link, cut or expose takes $\Oo(b(n)\cdot h(n))\in\Oo(\frac{(\log n)^{1+\epsilon}}{\epsilon\log\log n})\subseteq\oo((\log n)^{1+\epsilon})$ time.

The high-level algorithm makes at most a constant number of calls to link and cut for each insert or delete, so we are fine with the time for these operations.  However, we can no longer use $\Expose$ to implement $\Cover$, $\Uncover$, $\CoverLevel$ and $\MinCoveredEdge$, as that would take too long.

In this section, we will show how to overcome this limitation by working directly with the underlying tree.

\paragraph{The data}
The basic idea is to have each parent cluster store a \emph{buffer} for each of its children, containing all the $\cover$, $\cover^-$, $\cover^+$ and $\globalcover$ values.
Since the degree is $\Oo(w/\log\lmax)$, and each value uses at most $\Oo(\log\lmax)$ bits, these fit into a constant number of words, and so we can use standard bit tricks\footnote{See e.g.~\cite{DBLP:journals/jcss/FredmanW93} or~\cite{DBLP:journals/iandc/AlbersH97}.} to operate on the buffers for all children of a node in parallel. We will show how to implement $\Cover$, $\Uncover$, $\CoverLevel$, and $\MinCoveredEdge$, such that each of them only touches $\Oo(h(n))$ nodes and the buffers stored in those nodes.

Let $C$ be a cluster with children $A_1,\ldots,A_k$. Since $k\leq w/\log\lmax$, we can define the following vectors that each fit into a constant number of words.
\begin{align*}
  \packedcover_C &:= (\cover_{A_i})_{\set{1\leq i\leq k}}
  \\
  \packedcover^-_C &:= (\cover^-_{A_i})_{\set{1\leq i\leq k}}
  \\
  \packedcover^+_C &:= (\cover^+_{A_i})_{\set{1\leq i\leq k}}
  \\
  \packedglobalcover_C &:= (\globalcover_{A_i})_{\set{1\leq i\leq k}}
\end{align*}

The description of $\Split$ and $\Merge$ from Section~\ref{sec:coverlevel} still apply, if we think of the ``packed'' values as a separate layer of degree $1$ clusters between each pair of ``real'' clusters.

For concreteness, let $C$ be a cluster with children $A_1,\ldots,A_k$, and define operations
\begin{itemize}
\item $\CleanToBuffer(C)$. For each $1\leq i\leq k$: If $A_i$ is a path child of $C$ and\\ \mbox{$\max\set{\smash[b]{\packedcover_{C,i},\packedcover^-_{C,i}}}\leq\cover^-_C$}, set:
  \begin{align*}
    \packedcover^-_{C,i} &= \cover^-_C
\intertext{Then if $\packedcover_{C,i}\leq\max\set{\cover^-_C,\cover^+_C}$ set}
    \packedcover_{C,i} &= \cover^+_C
    \\
    \packedcover^+_{C,i} &= \cover^+_C
  \end{align*}
  After updating all $k$ children, set $\cover^-_C=\cover^+_C=-1$.
  Note that this can be done in parallel for all $1\leq i\leq k$ in
  constant time using bit tricks.
\item $\CleanToChild(C,i)$. If $A_i$ is a path child of $C$ and $\max\set{\smash[b]{\cover_{A_i},\cover^-_{A_i}}}\leq\packedcover^-_{C,i}$, set
  \begin{align*}
    \cover^-_{A_i} &= \packedcover^-_{C,i}
\intertext{Then if $\cover_{A_i}\leq\max\set{\smash[b]{\packedcover^-_{C,i},\packedcover^+_{C,i}}}$ set}
    \cover_{A_i} &= \packedcover^+_{C,i}
    \\
    \cover^+_{A_i} &= \packedcover^+_{C,i}
  \end{align*}
Finally set $\packedcover^-_{C,i}=\packedcover^+_{C,i}=-1$. Again, note that this takes constant time.
\item $\ComputeFromChild(C,i)$. Set
  \begin{align*}
    \packedcover_{C,i} &= \cover_{A_i}
    \\
    \packedcover^-_{C,i} &= -1
    \\
    \packedcover^+_{C,i} &= -1
    \\
    \packedglobalcover_{C,i} &= \globalcover_{A_i}
  \end{align*}

\item $\ComputeFromBuffer(C)$. For $1\leq i\leq k$ define
\begin{align*}
  \packedglobalcover'_{C,i}
  &=
  \begin{cases}
    \packedglobalcover_{C,i}
    &\text{if }\boundary A_i\subseteq \pi(C)
    \\
    &\text{or }\packedglobalcover_{C,i}\leq\packedcover_{C,i}
    \\
    \packedcover_{C,i}
    &\text{otherwise}
  \end{cases}
  \\
  \minglobaledge'_{C,i}
  &=
  \begin{cases}
    \mathrlap{\minglobaledge_{A_i}}
    \hphantom{\packedglobalcover_{C,i}}
      &\text{if }\boundary A_i\subseteq \pi(C)
      \\
      &\text{or }\globalcover_{A_i}\leq\cover_{A_i}
    \\
    \minpathedge_{A_i}
      &\text{otherwise}
  \end{cases}
\end{align*}
We can then compute the data for $C$ from the buffer as follows:
\begin{align*}
  \cover_C &=
  \begin{cases}
    \mathrlap{
    \min\limits_{\substack{
        1\leq i<k
        \\
        \boundary A_i\subseteq \pi(C)
    }} \packedcover_{C,i}
    }
    \hphantom{
    \;\text{where }j=\argmin\limits_{\substack{
      1\leq i<k
      \\
      \boundary A_i\subseteq \pi(C)
    }} \packedcover_{C,i}
    }
    &\text{if $\abs{\boundary C}=2$}
    \\
    \lmax&\text{otherwise}
  \end{cases}
  \\
  \minpathedge_C &=
  \begin{cases}
    \minpathedge_{A_j}
    &\text{if $\abs{\boundary C}=2$}
    \\
    \;\text{where }j=\argmin\limits_{\substack{
      1\leq i<k
      \\
      \boundary A_i\subseteq \pi(C)
    }} \packedcover_{C,i}
    \\
    \nil&\text{otherwise}
  \end{cases}
  \\
  \globalcover_C &= \min_{1\leq i<k}\packedglobalcover'_{C,i}
  \\
  \minglobaledge_C &= \minglobaledge'_{C,j}
  \\&\text{where }j=\argmin_{1\leq i<k} \packedglobalcover'_{C,i}
  \\
  \cover^-_C &= -1
  \\
  \cover^+_C &= -1
\end{align*}
This can be computed in constant time, because $(\packedglobalcover'_{C,i})_{\set{1\leq i\leq k}}$ fits into a constant number of words that can be computed in constant time using bit tricks, and thus each ``$\min$'' or ``$\argmin$'' is taken over values packed into a constant number of words.
\end{itemize}
Then $\Split(C)$ can be implemented by first calling $\CleanToBuffer(C)$, and then for each $1\leq i\leq k$ calling $\CleanToChild(C,i)$.  This ensures that all the lazy cover information is propagated down correctly. Similarly, $\Merge(C; A_1,\ldots,A_k)$ can be implemented by first calling $\ComputeFromChild(C,i)$ for each $1\leq i\leq k$, and then calling $\ComputeFromBuffer(C)$.  Thus $\Split$ and $\Merge$ each take $\Oo(b(n))$ time.

\paragraph{Computing $\CoverLevel(v)$ and $\MinCoveredEdge(v)$}
With the data described in the previous section, we can now answer the ``global'' queries as follows
\begin{align*}
  \CoverLevel(v) &= \globalcover_C
  \\
  \MinCoveredEdge(v) &= \minglobaledge_C
  \\
  &\hspace{-1cm}\text{where $C$ is the point cluster returned by $\root(v)$}
\end{align*}
Note that, for simplicity, we assume the top tree always has a single vertex exposed. This can easily be arranged by a constant number of calls to $\Expose$ after each link or cut, without affecting the asymptotic running time.  Computing $\CoverLevel(v)$ or $\MinCoveredEdge(v)$ therefore takes $\Oo(h(n))$ worst case time.

\paragraph{Computing $\CoverLevel(v,w)$ and $\MinCoveredEdge(v,w)$}
Since we can no longer use $\Expose$ to implement $\Cover$ and $\Uncover$, we need a little more machinery.

What saves us is that all the information we need to find $\CoverLevel(v,w)$ is stored in the $\Oo(h(n))$ clusters that have $v$ or $w$ as internal vertices, and that once we have that, we can find a single child $X$ of one of these clusters such that $\MinCoveredEdge(v,w)=\minpathedge_X$.

Before we get there, we have to deal with the complication of $\cover^-$ and $\cover^+$.  Fortunately, all we need to do is make $\Oo(h(n))$ calls to $\CleanToBuffer$ and $\CleanToChild$, starting from the root and going down towards $v$ and $w$.  Since each of these calls take constant time, we use only $\Oo(h(n))$ time on cleaning.

Now%
, the path $v\cdots w$ consists of $\Oo(h(n))$ edge-disjoint fragments, such that:
\begin{itemize}
\item Each fragment $f$ is associated with, and contained in, a single cluster $C_f$ whose parent has $v$ or $w$ as an internal vertex.
\item For each fragment $f$, the endpoints are either in $\set{v,w}$ (and then $C_f$ is a base cluster) or are boundary vertices of children of $C_f$.
\end{itemize}

We can find the fragments in $\Oo(h(n))$ time, and for each fragment $f$, we can in constant time find its cover level by examining $\packedcover_{C_f}$.

Let $f_1,\ldots,f_k$ be the fragments of the path, and for $1\leq i\leq k$ let $v_i,w_i$ be the endpoints of the fragment closest to $v,w$ respectively. Then\footnote{Recall that a \emph{path child} of $C$ is defined as a child that contains at least one edge of $\pi(C)$.}
\begin{align*}
  \CoverLevel(v,w) &= \min_{1\leq i\leq k}\CoverLevel(v_i,w_i)
  \\
  \MinCoveredEdge(v,w) &= \MinCoveredEdge(v_j,w_j)
  \\
  &\text{where }j=\argmin_{1\leq i\leq k}\CoverLevel(v_i,w_i)
  \\
  \MinCoveredEdge(v_j,w_j) &= \minpathedge_X
  \\
  &\text{where }X=\argmin_{Y\text{ path child of }C_{f_j}}\cover_Y
\end{align*}

So computing $\CoverLevel(v,w)$ or $\MinCoveredEdge(v,w)$ takes $\Oo(h(n))$ worst case time.

\paragraph{Cover and Uncover}
We are now ready to handle $\Cover(v,w,i)$ and $\Uncover(v,w,i)$.  First we make $\Oo(h(n))$ calls to $\CleanToBuffer$ and $\CleanToChild$. Then let $f_1,\ldots,f_k$ be the fragments of the $v\cdots w$ path, and for $1\leq i\leq k$ let $v_i,w_i$ be the endpoints of the fragment closest to $v,w$ respectively. Then for each $f\in f_1,\ldots,f_k$, and each path child $A_j$ of $C_f$, $\Cover(v,w,i)$ needs to set
\begin{align*}
  \packedcover_{C_f,j} &= \max\set{\packedcover_{C_f,j}, i}
  \\
  \packedcover^+_{C_f,j} &= \max\set{\packedcover^+_{C_f,j}, i}
\end{align*}
Similarly, for each $f\in f_1,\ldots,f_k$, and for each path child $A_j$ of $C_f$,
if $\packedcover_{C_f,j}\leq i$, $\Uncover(v,w,i)$ needs to set
\begin{align*}
  \packedcover_{C_f,j} &= -1
  \\
  \packedcover^+_{C_f,j} &= -1
  \\
  \packedcover^-_{C_f,j} &= \max\set{\packedcover^-_{C_f,j}, i}
\end{align*}
In each case, we can use bit tricks to make this take constant time per fragment. Finally, we need to update all the $\Oo(h(n))$ ancestors to the clusters we just changed.  We can do this bottom-up using $\Oo(h(n))$ calls to $\ComputeFromChild$ and $\ComputeFromBuffer$.

We conclude that $\Cover(v,w,i)$ and $\Uncover(v,w,i)$ each take worst case $\Oo(h(n))$ time.

\paragraph{Analysis} Choosing any $b(n)\in\Oo(w/\log\lmax)$ we get height $h(n)\in\Oo(\frac{\log n}{\log b(n)})$, so Link and Cut take worst case $\Oo(\frac{b(n)\log n}{\log b(n)})$ time with this CoverLevel structure.  The remaining operations, Connected, Cover, Uncover, CoverLevel and MinCoveredEdge all take $\Oo(\frac{\log n}{\log b(n)})$ worst case time.  For the purpose of our main result, choosing $b(n)\in\Theta(\sqrt{\log n})$ is sufficient.  Each cluster uses $\Oo(1)$ space, so the total space used is $\Oo(n)$.

\section{Saving space}\label{sec:fat}
We now apply the space-saving trick from Lemma~\ref{lem:savespace} to the FindSize structures from Section~\ref{sec:findsize} and~\ref{sec:ApproxCount}.
Let $D$ be the number of words used for each size vector in our FindSize structure.  This is $\Oo(\log n)$ for the purely combinatorial version, and $\Oo(\log\log n)$ in the version using approximate counting.  As shown previously these use $s(n)=\Oo(D)$ space per cluster and $t(n)=\Oo(\log n\cdot D)$ worst case time per merge and split.

\begin{lemma}
  The complete information for a cluster of size $q$ in the FindSize structure, including information that would be shared with its children, has total size $s_0(q,n)=\Oo(q+\lmax\cdot D)$.
\end{lemma}
\begin{proof}
  The complete information for a cluster $C$ with $\sizeof{C}=q$ consists of 
  \begin{itemize}
  \item $c(e)$ for all $e\in C$.
  \item $\cover_C$, $\cover^-_C$, $\cover^+_C$, $\globalcover_C$, $\size_C$.
  \item $\partsize_{C,v,i}$ and $\diagsize_{C,v,i}$ for $v\in\boundary C$ and $-1\leq i\leq\lmax$.
  \end{itemize}
  The total size for all of these is $s_0(q,n)=\Oo(q+\lmax\cdot D)$
\end{proof}
Note that when keeping $n$ fixed, this is clearly $\Oo(q)$.  In particular, we can choose $q(n)\in\Theta(\lmax\cdot D)$ such that $s(n)<s_0(q(n),n)\in\Oo(q(n))$.

\begin{lemma}
  The complete information for a cluster of size $q$ in the FindSize structure, including information that would be shared with its children, can be computed directly in time $t_0(q,n)=\Oo(q\log q+\lmax\cdot D)$.
\end{lemma}
\begin{proof}
  Let $C$ be the cluster of size $\sizeof{C}=q$. For each $v\in\boundary C$, we can in $\Oo(q)$ time find and partition the cluster path into the at most $\lmax$ parts such that in part $i$, each vertex $m$ on the cluster path have $\CoverLevel(v,m)=i$.
For each part $i$, run the following algorithm:
\begin{algorithmic}[1]
  \State Vector $x\gets \vec{0}$
  \State Initialize empty max-queue $Q$
  \State $j\gets\lmax$
  \For{$w\gets$ each vertex in the fragment that is on $\pi(C)$}
    \State Mark $w$ as visited
    \State $x_j\gets x_j+1$
    \For{$e\gets$ each edge incident to $w$ that is not on $\pi(C)$}
      \If{$c(e)\geq 0$}
        \State Add $e$ to $Q$ with key $c(e)$
      \EndIf
    \EndFor
  \EndFor
  \While{$Q$ is not empty}
    \State $e\gets\Call{extract-max}{Q}$
    \While{$c(e)<j$}
      \State $x_{j-1}=x_j$
      \State $j\gets j-1$
    \EndWhile
    \State $w\gets$ the unvisited vertex at the end of $e$
    \State Mark $w$ as visited
    \State $x_j\gets x_j+1$
    \For{$e\gets$ each edge incident to $w$ that has an unvisited end}
      \If{$c(e)\geq 0$}
        \State Add $e$ to $Q$ with key $c(e)$
      \EndIf
    \EndFor
  \EndWhile
  \State $\partsize_{C,v,i}\gets x$
  \State $\diagsize_{C,v,i}\gets M\!(i)\cdot x$
\end{algorithmic}
If the $i$th part has size $q_i$ than it can be processed this way in $\Oo(q_i\log q_i+D)$ time.  Summing over all $\Oo(\lmax)$ parts gives the desired result.
\end{proof}

\paragraph{Analysis}
Applying Lemma~\ref{lem:savespace} with the $s(n)$, $t(n)$, $s_0(q,n)$, $t_0(q,n)$ and $q(n)$ derived in this section immediately gives a FindSize structure with $\Oo(\log n\cdot D \cdot \log\lmax)$ worst case time per operation and using $\Oo(n)$ space.  A completely analogous argument shows that we can convert the bitpacking-free version of the FindFirstLabel structure from $\Oo(\log n\cdot\lmax\cdot\log\lmax)$ time and $\Oo(m+n\cdot\lmax)$ space to one using linear space. (If bitpacking is allowed the structure already used linear space). In either case is the same time per operation as the original versions, so using the modified version here does not affect the overall running time, but reduces the total space of each bridge-finding structure to $\Oo(m+n)$.

Note that we can explicitly store lists with all the least-covered edges for these large base clusters, so this does not change the time to report the first $k$ least-covered edges.

\bibliographystyle{plain}
\bibliography{paper}

\begin{thebibliography}{10}

\bibitem{DBLP:journals/iandc/AlbersH97}
Susanne Albers and Torben Hagerup.
\newblock Improved parallel integer sorting without concurrent writing.
\newblock {\em Inf. Comput.}, 136(1):25--51, 1997.

\bibitem{Alstrup:2005}
Stephen Alstrup, Jacob Holm, Kristian~De Lichtenberg, and Mikkel Thorup.
\newblock Maintaining information in fully dynamic trees with top trees.
\newblock {\em ACM Trans. Algorithms}, 1(2):243--264, October 2005.

\bibitem{BIEDL2001110}
Therese~C. Biedl, Prosenjit Bose, Erik~D. Demaine, and Anna Lubiw.
\newblock Efficient algorithms for petersen's matching theorem.
\newblock {\em Journal of Algorithms}, 38(1):110 -- 134, 2001.

\bibitem{Diks2010}
Krzysztof Diks and Piotr Stanczyk.
\newblock {\em Perfect Matching for Biconnected Cubic Graphs in $O(n \log^2 n)$
  Time}, pages 321--333.
\newblock Springer Berlin Heidelberg, Berlin, Heidelberg, 2010.

\bibitem{Eppstein93improvedsparsification}
David Eppstein, Zvi Galil, and Giuseppe~F. Italiano.
\newblock Improved sparsification.
\newblock Technical report, 1993.

\bibitem{DBLP:journals/siamcomp/Frederickson85}
Greg~N. Frederickson.
\newblock Data structures for on-line updating of minimum spanning trees, with
  applications.
\newblock {\em {SIAM} Journal on Computing}, 14(4):781--798, 1985.

\bibitem{DBLP:journals/siamcomp/Frederickson97}
Greg~N. Frederickson.
\newblock Ambivalent data structures for dynamic 2-edge-connectivity and k
  smallest spanning trees.
\newblock {\em {SIAM} J. Comput.}, 26(2):484--538, 1997.

\bibitem{DBLP:journals/jcss/FredmanW93}
Michael~L. Fredman and Dan~E. Willard.
\newblock Surpassing the information theoretic bound with fusion trees.
\newblock {\em J. Comput. Syst. Sci.}, 47(3):424--436, 1993.

\bibitem{Gabow:2001}
Harold~N. Gabow, Haim Kaplan, and Robert~Endre Tarjan.
\newblock Unique maximum matching algorithms.
\newblock {\em J. Algorithms}, 40(2):159--183, 2001.
\newblock Announced at STOC '99.

\bibitem{DBLP:journals/corr/GibbKKT15}
David Gibb, Bruce~M. Kapron, Valerie King, and Nolan Thorn.
\newblock Dynamic graph connectivity with improved worst case update time and
  sublinear space.
\newblock {\em CoRR}, abs/1509.06464, 2015.

\bibitem{Henzinger1997}
Monika~R. Henzinger and Valerie King.
\newblock {\em Maintaining minimum spanning trees in dynamic graphs}, pages
  594--604.
\newblock Springer Berlin Heidelberg, Berlin, Heidelberg, 1997.

\bibitem{Henzinger97fullydynamic}
Monika~Rauch Henzinger and Valerie King.
\newblock Fully dynamic 2-edge connectivity algorithm in polylogarithmic time
  per operation, 1997.

\bibitem{Holm:1998}
Jacob Holm, Kristian de~Lichtenberg, and Mikkel Thorup.
\newblock Poly-logarithmic deterministic fully-dynamic algorithms for
  connectivity, minimum spanning tree, 2-edge, and biconnectivity.
\newblock In {\em Proceedings of the Thirtieth Annual ACM Symposium on Theory
  of Computing}, STOC '98, pages 79--89, New York, NY, USA, 1998. ACM.

\bibitem{Holm:2001}
Jacob Holm, Kristian de~Lichtenberg, and Mikkel Thorup.
\newblock Poly-logarithmic deterministic fully-dynamic algorithms for
  connectivity, minimum spanning tree, 2-edge, and biconnectivity.
\newblock {\em J. ACM}, 48(4):723--760, July 2001.

\bibitem{Huang:2017}
Shang-En Huang, Dawei Huang, Tsvi Kopelowitz, and Seth Pettie.
\newblock Fully dynamic connectivity in o(log n(log log n)2) amortized expected
  time.
\newblock In {\em Proceedings of the Twenty-Eighth Annual ACM-SIAM Symposium on
  Discrete Algorithms}, SODA '17, pages 510--520, Philadelphia, PA, USA, 2017.
  Society for Industrial and Applied Mathematics.

\bibitem{Kapron:2013}
Bruce~M. Kapron, Valerie King, and Ben Mountjoy.
\newblock Dynamic graph connectivity in polylogarithmic worst case time.
\newblock In {\em Proceedings of the Twenty-fourth Annual ACM-SIAM Symposium on
  Discrete Algorithms}, SODA '13, pages 1131--1142, Philadelphia, PA, USA,
  2013. Society for Industrial and Applied Mathematics.

\bibitem{10.2307/2325085}
Donald~E. Knuth.
\newblock Two notes on notation.
\newblock {\em The American Mathematical Monthly}, 99(5):403--422, 1992.

\bibitem{Knuth:1998}
Donald~E. Knuth.
\newblock {\em The Art of Computer Programming, Volume 3: (2nd Ed.) Sorting and
  Searching}.
\newblock Addison Wesley Longman Publishing Co., Inc., Redwood City, CA, USA,
  1998.

\bibitem{Kotzig59}
Anton Kotzig.
\newblock {\em On the theory of finite graphs with a linear factor II.}
\newblock 1959.

\bibitem{menger1927allgemeinen}
Karl Menger.
\newblock {Zur allgemeinen Kurventheorie}.
\newblock {\em Fundamenta Mathematicae}, 10, 1927.

\bibitem{patrascu2006logarithmic}
Mihai Patrascu and Erik~D Demaine.
\newblock Logarithmic lower bounds in the cell-probe model.
\newblock {\em SIAM Journal on Computing}, 35(4):932--963, 2006.

\bibitem{petersen1891}
Julius Petersen.
\newblock Die {T}heorie der regulären graphs.
\newblock {\em Acta Math.}, 15:193--220, 1891.

\bibitem{Thorup:2000}
Mikkel Thorup.
\newblock Near-optimal fully-dynamic graph connectivity.
\newblock In {\em Proceedings of the Thirty-second Annual ACM Symposium on
  Theory of Computing}, STOC '00, pages 343--350, New York, NY, USA, 2000. ACM.

\bibitem{Wulff-Nilsen16}
Christian Wulff{-}Nilsen.
\newblock Faster deterministic fully-dynamic graph connectivity.
\newblock In {\em Encyclopedia of Algorithms}, pages 738--741. 2016.

\end{thebibliography}

\appendix

\section{Details of the high level algorithm}\label{app:highlevel}

\lemhighlevel*
\begin{proof}
  The only part of the high level algorithm from~\cite{Holm:2001} that
  does not directly and trivially translate into a call of the
  required dynamic tree operations (see pseudocode below) is in the Swap method where given a
  tree edge $e=(v,w)$ we need to find a nontree edge $e'$ covering $e$
  with $\ell(e')=i=\CoverLevel(e)$.  We can find this
  $e'$ by using FindFirstLabel and increasing the level of each
  non-tree edge we examine that does not cover $e$.  For at least one
  side of $(v,w)$, all non-tree edges at level $i$ incident to that
  side will either cover $e$ or can safely have their level increased
  without violating the size invariant.  So we can simply search the side where the level $i$ component is smallest until we find the required edge (which must exist since $e$ was covered on level $i$). 
  The amortized cost of all operations remain
  unchanged with this implementation. Counting the number of operations (see Table~\ref{tbl:opcounts}) gives the desired bound.
\end{proof}

\begin{table}[htb]
    \center
    \begin{adjustbox}{
            max width={\dimexpr(\paperwidth+2\textwidth)/3\relax},
            Trim={\dimexpr(\paperwidth-\textwidth)/6\relax} {0pt}
        }
        \small
        \begin{tabular}{|r|l||c|c|c|c|c|}
            \hline
            \multirow{2}{*}{\#}&\multirow{2}{*}{Operation} & \multicolumn{5}{c|}{\#Calls during} \\
            && Insert+Delete & FindBridge$(v)$ & FindBridge$(v,w)$ & Size(v) &  $2$-Size(v) \\
            \hline
            \ref{it:link}&Link$(v,w,e)$ & $1$ & $0$ & $0$ & $0$ & $0$ \\
            \ref{it:cut}&Cut$(e)$ & $1$ & $0$ & $0$ & $0$ & $0$ \\
            \hline
            \ref{it:conn}&Connected$(v,w)$ & $\log n$ & $0$ & $1$ & $0$ & $0$ \\
            \ref{it:cover}&Cover$(v,w,i)$ & $\log n$ & $0$ & $0$ & $0$ & $0$ \\
            \ref{it:uncover}&Uncover$(v,w,i)$ & $1$ & $0$ & $0$ & $0$ & $0$ \\
            \ref{it:cl1}&CoverLevel$(v)$ & $0$ & $1$ & $0$ & $0$ & $0$ \\
            \ref{it:cl2}&CoverLevel$(v,w)$ & $1$ & $0$ & $1$ & $0$ & $0$ \\
            \ref{it:mce1}&MinCoveredEdge$(v)$ & $0$ & $1$ & $0$ & $0$ & $0$ \\
            \ref{it:mce2}&MinCoveredEdge$(v,w)$ & $0$ & $0$ & $1$ & $0$ & $0$ \\
            \hline
            \ref{it:al}&AddLabel$(v,l,i)$ & $\log n$ & $0$ & $0$ & $0$ & $0$ \\
            \ref{it:rl}&RemoveLabel$(l)$ & $\log n$ & $0$ & $0$ & $0$ & $0$ \\
            \ref{it:ffl}&FindFirstLabel$(v,w,i)$ & $\log n$ & $0$ & $0$ & $0$ & $0$ \\
            \hline
            \multirow{2}{*}{\ref{it:fs}}
            &FindSize$(v,w,i)$ & $\log n$ & $0$ & $0$ & $0$ & $1$ \\
            &FindSize$(v,v,-1)$ & $0$ & $0$ & $0$ & $1$ & $0$ \\
            \hline
        \end{tabular}
    \end{adjustbox}
    \caption{\label{tbl:opcounts}Overview of how many times each tree operation is called for each graph operation, ignoring constant factors. The ``Insert+Delete'' column is amortized over any sequence starting with an empty set of edges.  The remaining columns are worst case.
    }
\end{table}

\begin{algorithmic}[1]
  \Function{$2$-edge-connected}{$v$, $w$}
  \State \Return \Call{T.Connected}{$v$, $w$} $\wedge$ \Call{T.CoverLevel}{$v$, $w$}$\geq 0$
  \EndFunction

  \Function{FindBridge}{$v$}
  \If{\Call{T.CoverLevel}{$v$}$=-1$}
    \State \Return \Call{T.MinCoveredEdge}{$v$}
  \Else
    \State \Return $\nil$
  \EndIf
  \EndFunction

  \Function{FindBridge}{$v$, $w$}
  \If{\Call{T.CoverLevel}{$v$, $w$}$=-1$}
    \State \Return \Call{T.MinCoveredEdge}{$v$, $w$}
  \Else
    \State \Return $\nil$
  \EndIf
  \EndFunction

  \Function{Size}{$v$}
  \State \Return \Call{T.FindSize}{$v$,$v$,$-1$}
  \EndFunction

  \Function{$2$-Size}{$v$}
  \State \Return \Call{T.FindSize}{$v$,$v$,$0$}
  \EndFunction

  \Function{Insert}{$v$, $w$, $e$}
  \If{$\neg$\Call{T.Connected}{$v$, $w$}}
    \State \Call{T.Link}{$v$, $w$, $e$}
    \State $\ell(e)\gets \lmax$
  \Else
    \State \Call{T.AddLabel}{$v$, $e.$label1, $0$}
    \State \Call{T.AddLabel}{$w$, $e.$label2, $0$}
    \State $\ell(e)\gets 0$
    \State \Call{T.Cover}{$v$, $w$, $0$}
  \EndIf
  \EndFunction

  \Function{Delete}{$e$}
  \State $(v,w)\gets e$
    \State $\alpha\gets\ell(e)$
    \If{$\alpha=\lmax$}
      \State $\alpha\gets\Call{T.CoverLevel}{v,w}$
      \If{$\alpha=-1$}
        \State $\Call{T.Cut}{e}$
        \State \Return
      \EndIf
      \State \Call{Swap}{$e$}
    \EndIf
    \State \Call{T.RemoveLabel}{$e.$label1}
    \State \Call{T.RemoveLabel}{$e.$label2}
    \State $\Call{T.Uncover}{v,w,\alpha}$
    \For{$i\gets \alpha,\ldots,0$}
      \State \Call{Recover}{$w$,$v$,$i$}
    \EndFor
  \EndFunction

  \Function{Swap}{$e$}
  \State $(v,w)\gets e$
  \State $\alpha\gets\Call{T.CoverLevel}{v,w}$
  \State $\Call{T.Cut}{e}$
  \State $e'\gets$\Call{FindReplacement}{$v$,$w$,$\alpha$}
  \State $(x,y)\gets e'$
  \State \Call{T.RemoveLabel}{$e'.$label1}
  \State \Call{T.RemoveLabel}{$e'.$label2}
  \State \Call{T.Link}{$x,y,e'$}
  \State $\ell(e')\gets \lmax$
  \State \Call{T.AddLabel}{$v$,$e.$label1, $\alpha$}
  \State \Call{T.AddLabel}{$w$,$e.$label2, $\alpha$}
  \State $\ell(e)\gets \alpha$
  \State \Call{T.Cover}{$v$,$w$,$\alpha$}
  \EndFunction

  \Function{FindReplacement}{$v$,$w$,$i$}
    \State $s_v\gets\Call{T.FindSize}{v,v,i}$
    \State $s_w\gets\Call{T.FindSize}{w,w,i}$
    \If{$s_v\leq s_w$}
      \State \Return $\Call{RecoverPhase}{v,v,i,s_v}$
    \Else
      \State \Return $\Call{RecoverPhase}{w,w,i,s_w}$
    \EndIf
  \EndFunction

  \Function{Recover}{$v$,$w$,$i$}
    \State $s\gets\floor{ \Call{T.FindSize}{v, w, i} / 2}$
    \State \Call{RecoverPhase}{$v$,$w$,$i$,$s$}
    \State \Call{RecoverPhase}{$w$,$v$,$i$,$s$}
  \EndFunction

  \Function{RecoverPhase}{$v,w,i,s$}
  \State $l\gets\Call{T.FindFirstLabel}{v,w,i}$
  \While{$l\neq\nil$}
    \State $e\gets l.$edge
    \State $(q,r)\gets e$
    \If{$\neg$\Call{T.Connected}{$q$, $r$}}
      \State \Return $e$
    \EndIf
    \If{$\Call{T.FindSize}{q,r,i+1}\leq s$}
        \State \Call{T.RemoveLabel}{$e.$label1}
        \State \Call{T.RemoveLabel}{$e.$label2}
        \State \Call{T.AddLabel}{$q$, $e.$label1, $i+1$}
        \State \Call{T.AddLabel}{$r$, $e.$label2, $i+1$}
        \State $\ell(e)=i+1$
        \State \Call{T.Cover}{$q$,$r$,$i+1$}
      \Else
        \State \Call{T.Cover}{$q$,$r$,$i$}
        \State \Return $\nil$
    \EndIf
    \State $l\gets\Call{T.FindFirstLabel}{v,w,i}$
  \EndWhile
  \State \Return $\nil$
  \EndFunction
\end{algorithmic}

\section{Pseudocode for the $\CoverLevel$ structure}\label{app:coverlevel}
\begin{algorithmic}[1]
  \Function{CL.Cover}{$v$,$w$,$i$}
  \State $C\gets$ \Call{TopTree.Expose}{$v$, $w$}
  \State $\cover_C\gets\max\set{\cover_C,i}$
  \State $\cover^+_C\gets\max\set{\cover^+_C,i}$
  \EndFunction

  \Function{CL.Uncover}{$v$,$w$,$i$}
  \State $C\gets$ \Call{TopTree.Expose}{$v$, $w$}
  \State $\cover_C\gets-1$
  \State $\cover^+_C\gets-1$
  \State $\cover^-_C\gets\max\set{\cover^-_C,i}$
  \EndFunction

  \Function{CL.CoverLevel}{$v$}
  \State $C\gets$ \Call{TopTree.Expose}{$v$}
  \State \Return $\globalcover_C$
  \EndFunction

  \Function{CL.CoverLevel}{$v$, $w$}
  \State $C\gets$ \Call{TopTree.Expose}{$v$, $w$}
  \State \Return $\cover_C$
  \EndFunction

  \Function{CL.MinCoveredEdge}{$v$}
  \State $C\gets$ \Call{TopTree.Expose}{$v$}
  \State \Return $\minglobaledge_C$
  \EndFunction

  \Function{CL.MinCoveredEdge}{$v$, $w$}
  \State $C\gets$ \Call{TopTree.Expose}{$v$, $w$}
  \State \Return $\minpathedge_C$
  \EndFunction

  \Function{CL.Split}{$C$}
  \For{each path child $D$ of $C$}
    \If{$\max\set{\cover_D,\cover^-_D}\leq\cover^-_C$}
      \State $\cover^-_D \gets \cover^-_C$
    \EndIf
    \If{$\cover_D\leq\max\set{\cover^-_D,\cover^+_D}$}
      \State $\cover_D \gets \cover^+_C$
      \State $\cover^+_D \gets \cover^+_C$
    \EndIf
  \EndFor
  \EndFunction

  \Function{CL.Merge}{$C$; $A_1,\ldots,A_k$}
  \State $\cover_C \gets \lmax$
  \State $\minpathedge_C \gets \nil$
  \State $\globalcover_C \gets \lmax$
  \State $\minglobaledge_C \gets \nil$
  \For{$i\gets 1,\ldots,k$}
    \If{$\boundary A_i\subseteq\pi(C)$}
      \If{$\cover_{A_i}<\cover_C$}
        \State $\cover_C \gets \cover_{A_i}$
        \State $\minpathedge_C \gets \minpathedge_{A_i}$
      \EndIf
    \Else
      \If{$\cover_{A_i}<\globalcover_C$}
        \State $\globalcover_C \gets \cover_{A_i}$
        \State $\minglobaledge_C \gets \minpathedge_{A_i}$
      \EndIf
    \EndIf
    \If{$\globalcover_{A_i}<\globalcover_C$}
      \State $\globalcover_C \gets \globalcover_{A_i}$
      \State $\minglobaledge_C \gets \minglobaledge_{A_i}$
    \EndIf
  \EndFor
  \State $\cover^-_C \gets -1$
  \State $\cover^+_C \gets -1$
  \EndFunction

  \Function{CL.Create}{$C$; edge $e$}
  \State $\cover_C \gets -1$
  \State $\globalcover_C \gets -1$
  \If{$C$ is a point cluster}
    \State $\minpathedge_C \gets \nil$
    \State $\minglobaledge_C \gets e$
  \Else
    \State $\minpathedge_C \gets e$
    \State $\minglobaledge_C \gets \nil$
  \EndIf
  \State $\cover^-_C \gets -1$
  \State $\cover^+_C \gets -1$
  \EndFunction
\end{algorithmic}

\section{Pseudocode for the $\FindSize$ structure}\label{app:findsize}
In the following, we use the notation
\begin{align*}
  [\operatorname{key}: \partsize,\diagsize]
\end{align*}
to denote the root of a new tree consisting of a single node with the given values.
And for a given tree root and given $x,y$
\begin{align*}
  (\tree_{\set{x\leq i\leq y}})
\end{align*}
is the root of the subtree consisting of all nodes whose keys are in the given range.
Similarly, for any given $i$, let
\begin{align*}
  (\tree_{i})
\end{align*}
denote the node in the tree having the given key.

\begin{algorithmic}[1]
  \Function{FS.FindSize}{$v$, $w$, $i$}
  \State $C\gets$ \Call{TopTree.Expose}{$v$, $w$}
  \State \Return $\size_{C,i}$
  \EndFunction

  \Function{FS.Merge}{$C$; $A$, $B$}
  \State $\set{c}\gets\boundary A\cap\boundary B$
  \If{$c\in\pi(C)$}
    \Comment{Merge along path}
    \If{$\abs{\boundary C}<=1$}
      \State $a\gets c$, $b\gets c$
    \Else
      \State $\set{a,b}\gets\boundary C$ with $a\in\boundary A$ and $b\in\boundary B$.
    \EndIf
    \State $\size_C\gets\size_A+\size_B$
    \For{$(x,X)\gets (a,A),(b,B)$}
      \If{$x=c$}
        \State $\tree'_{X,x}\gets\tree_{X,x}$, $\undo'_{X,x}\gets\nil$
      \Else
        \For{$v\gets x, c$}
          \State $\ell\gets\max\set{\cover^-_X,\cover^+_X}$
          \State $s\gets (\tree_{X,v}).\partsizesum$
          \State $d\gets M\!(\cover^+_X)*s$
          \State $\tree'_{X,v}\gets\tree_{X,v,\set{i>\ell}}$, $\undo'_{X,v}\gets\tree_{X,v,\set{i\leq\ell}}$
          \State $\tree'_{X,v}\gets\tree'_{X,v}+[\cover^+_X:s,d]$
        \EndFor
      \EndIf
    \EndFor
    \For{$(x,X,y,Y)\gets (a,A,b,B),(b,B,a,A)$}
      \State $s\gets (\tree'_{Y,c,\set{\cover_X\leq i\leq\lmax}}).\partsizesum$
      \State $p\gets(\tree'_{X,x,\cover_X}).\partsize+s$
      \State $d\gets(\tree'_{X,x,\cover_X}).\diagsize+M\!(\cover_X)*s$

      \If{$x=c$}
        \State $\tree''_{X,x}\gets [\lmax:\size_X,\size_X]$,
               $\undo''_{X,x}\gets \nil$
      \Else
        \State $\tree''_{X,x}\gets\tree'_{X,x,\set{i>\cover_X}}$,
               $\undo''_{X,x}\gets\tree'_{X,x,\set{i\leq\cover_X}}$
      \EndIf
      \If{$y=c$}
        \State $\tree'''_{Y,c}\gets\nil$,
               $\undo'''_{Y,c}\gets [\lmax:\size_Y,\size_Y]$
      \Else
        \State $\tree'''_{Y,c}\gets\tree'_{Y,c,\set{i<\cover_X}}$,
               $\undo'''_{Y,c}\gets\tree'_{Y,c,\set{i\geq\cover_X}}$
      \EndIf

      \State  $\tree_{C,x}\gets\tree''_{X,x}+[\cover_X:p,d]+\tree'''_{Y,c}$
    \EndFor
  \Else
    \Comment{Merge off path}
    \State $\set{a}\gets\boundary C\setminus\set{c}$
    \If{$a\not\in\boundary A$}
      \State Swap $A$ and $B$
    \EndIf
    \State $\ell\gets\max\set{\cover^-_A,\cover^+_A}$
    \State $d \gets (\tree_{A,a,\set{\ell<i\leq\lmax}}).\diagsizesum$
    \State $p \gets (\tree_{A,a,\set{-1\leq i\leq\ell}}).\partsizesum$
    \State $\size_C\gets d + M\!(\cover^+_A)*p + M\!(\cover_A)*\size_B$
    \State $\tree_{C,a} \gets [\lmax:\size_C,\size_C]$
  \EndIf
  \EndFunction

  \Function{FS.Split}{$C$}
  \State $A,B\gets$ the children of $C$
  \State $\set{c}\gets\boundary A\cap\boundary B$
  \If{$c\in\pi(C)$}
    \Comment{Split along path}
    \If{$\abs{\boundary C}<=1$}
      \State $a\gets c$, $b\gets c$
    \Else
      \State $\set{a,b}\gets\boundary C$ with $a\in\boundary A$ and $b\in\boundary B$.
    \EndIf
    \For{$(x,X,y,Y)\gets (a,A,b,B),(b,B,a,A)$}
      \State $\tree''_{X,x}\gets \tree_{C,x,\set{i>\cover_X}}$,
             $\tree'''_{Y,c}\gets \tree_{C,x,\set{i<\cover_X}}$
      \If{$y\neq c$}
        \State $\tree'_{Y,c}\gets\tree'''_{Y,c}+\undo'''_{Y,c}$
      \EndIf
      \If{$x\neq c$}
        \State $\tree'_{X,x}\gets\tree''_{X,x}+\undo''_{X,x}$
      \EndIf
    \EndFor
    \For{$(x,X)\gets (a,A),(b,B)$}
      \If{$x\neq c$}
        \For{$v\gets x,c$}
          \State $\tree_{X,v}\gets\tree'_{X,v,\set{i>\cover^+_X}}+\undo'_{X,v}$
        \EndFor
      \EndIf
    \EndFor
  \EndIf
  \EndFunction

  \Function{FS.Create}{$C$; edge $e$}
  \State $\size_C\gets \vec{0}$
  \For{$v\in\boundary C$}
    \State $\tree_{C,v}\gets [\lmax: \vec{0},\vec{0}]$
  \EndFor
  \EndFunction

  \Function{FS.Create}{$C$; vertex label $l$}
  \State $\size_C\gets (1)_{\set{0\leq i<\lmax}}$
  \For{$v\in\boundary C$}
    \State $\tree_{C,v} = [\lmax:\size_C,\size_C]$
  \EndFor
  \EndFunction
\end{algorithmic}

\end{document}